\documentclass{article}

\PassOptionsToPackage{numbers,compress}{natbib}
\usepackage[main,preprint]{neurips_2026}

\usepackage[utf8]{inputenc}
\usepackage[T1]{fontenc}
\usepackage{amsmath,amssymb,amsthm,mathtools,bm}
\usepackage{booktabs,multirow,array,tabularx}
\usepackage{enumitem}
\usepackage{graphicx}
\usepackage{float}
\usepackage{microtype}
\usepackage{xcolor}
\usepackage{xspace}
\usepackage{url}
\definecolor{linkblue}{rgb}{0.12,0.30,0.65}
\usepackage[pagebackref,breaklinks,colorlinks,citecolor=linkblue,urlcolor=linkblue,linkcolor=linkblue]{hyperref}
\hypersetup{pdftitle={Can Pixels Alone Reveal Image Origin? Minimax Limits and Learnable Interfaces for Passive Provenance},pdfauthor={Kai Yao}}
\usepackage[capitalise,noabbrev]{cleveref}

\newcommand{\Description}[1]{}

\newtheorem{theorem}{Theorem}
\newtheorem{lemma}[theorem]{Lemma}
\newtheorem{corollary}[theorem]{Corollary}
\newtheorem{proposition}[theorem]{Proposition}

\crefname{theorem}{Theorem}{Theorems}
\Crefname{theorem}{Theorem}{Theorems}
\crefname{lemma}{Lemma}{Lemmas}
\Crefname{lemma}{Lemma}{Lemmas}
\crefname{corollary}{Corollary}{Corollaries}
\Crefname{corollary}{Corollary}{Corollaries}
\crefname{proposition}{Proposition}{Propositions}
\Crefname{proposition}{Proposition}{Propositions}
\crefname{definition}{Definition}{Definitions}
\Crefname{definition}{Definition}{Definitions}
\crefname{remark}{Remark}{Remarks}
\Crefname{remark}{Remark}{Remarks}
\crefname{example}{Example}{Examples}
\Crefname{example}{Example}{Examples}

\newcommand{\X}{\mathcal{X}}
\newcommand{\Y}{\mathcal{Y}}
\newcommand{\Qset}{\mathcal{Q}}
\newcommand{\Kset}{\mathcal{K}}
\newcommand{\TV}{\mathrm{TV}}
\newcommand{\Acc}{\mathrm{Acc}}
\newcommand{\Gap}{\mathrm{Gap}}
\newcommand{\E}{\mathbb{E}}
\newcommand{\Prb}{\mathbb{P}}
\newcommand{\one}{\mathbf{1}}
\newcommand{\supp}{\mathrm{supp}}
\newcommand{\R}{\mathbb{R}}
\newcommand{\U}{\mathcal{U}}
\newcommand{\AttackRad}{\Delta}
\newcommand{\Human}{\mathsf{H}}
\newcommand{\AI}{\mathsf{A}}

\DeclareMathOperator*{\argmax}{arg\,max}
\DeclareMathOperator*{\argmin}{arg\,min}

\title{Can Pixels Alone Reveal Image Origin?\\Minimax Limits and Learnable Interfaces for Passive Provenance}

\author{Kai Yao\\
School of Informatics\\
University of Edinburgh\\
\texttt{kai.yao@ed.ac.uk}}

\begin{document}

\maketitle

\begin{center}
\small
This work has been accepted for publication in the proceedings of the 40th Annual Conference on Neural Information Processing Systems (NeurIPS 2026).
\end{center}

\begin{abstract}
Passive image provenance asks whether pixels alone can reveal where an image came from: a human, an aggregate AI class, or a particular generator. This becomes a robustness problem once a source image can be edited before the verifier sees it. We study the problem as source--target verification under adversarial distribution shift. Our first result gives the exact best-case limit for any image-only verifier: the largest robust target-acceptance gap equals the minimum total-variation distance between the target distribution and the set of attacked source distributions. This quantity depends on the source, target, and edit class, not on the verifier architecture. Our second result explains why deployed public verifiers can fail before this statistical limit is reached. If the verifier can be emulated on the attack region to error $\varepsilon$, then a surrogate black-box attack reaches target acceptance within $2\varepsilon$ plus optimization error of the white-box optimum; score-revealing logistic and softmax heads over public features are identifiable, and approximate score access gives stable recovery bounds. A finite-state experiment checks the minimax identity where both sides are computable. On same-prompt real/diffusion benchmarks, the evaluated public CLIP verifiers fail under targeted pixel attacks, while a ResNet-18 victim exhibits partial fake-to-real transfer. Binary feedback with abstention reduces measured attack success, but positive empirical gap upper bounds do not establish robustness. These results motivate separate evaluation of the source--target statistical ceiling and the information released by a deployed verifier.
\end{abstract}

\section{Introduction}

A passive image-provenance system receives an image and must infer its origin from pixels alone, continuing a forensic line that uses sensor and camera-model traces~\cite{lukas2006digital,cozzolino2019noiseprint} and learned forensic embeddings~\cite{cozzolino2018forensictransfer}. It may ask whether the image is real or AI-generated~\cite{wang2020cnnspot,ojha2023universal,wang2023dire,cozzolino2024clip}, or it may ask which generator produced it~\cite{yu2019attributing,yang2023pose,song2024manifpt}. This looks like a classification problem on clean data. Under adaptive editing, detector-evasion and adversarial-example results show that image edits can alter a classifier's decision~\cite{goodfellow2015explaining,papernot2017practical,carlini2017towards,madry2017towards,athalye2018obfuscated,carlini2020deepfake,hussain2021adversarial,wesselkamp2022misleading,yao2025smudged}. We focus on targeted forgery: an edited source image is intended to be accepted as a chosen target.

We use the following source--target view throughout the paper. A source image is drawn from a source distribution $P_s$. The attacker may edit it within an admissible set $\Gamma(x)$, where $x$ is the original image. The verifier outputs a target score $f_t(x)\in[0,1]$, interpreted as acceptance of target label $t$. The central question is: how much larger can the target score be on true target images than on any edited source images?

There are two distinct reasons this separation can fail. The first reason is statistical. If the edited source distribution can get close to the target distribution, then no image-only verifier can reliably separate them. The second reason is interface-based. Even when image-level signal remains, a particular public verifier may reveal enough information for an attacker to learn a surrogate and optimize against it. These two failures look similar in experiments, but they require different fixes: better training cannot beat the statistical ceiling, while interface design can reduce learnability of the deployed verifier.

\paragraph{Main take-away.}
The paper separates these two mechanisms. Barrier~I is an exact minimax ceiling: the best possible passive verifier is limited by the total-variation distance from the target distribution to the attacked source family. Barrier~II is a deployed-verifier failure mode: a learnable public interface can make black-box attacks nearly white-box. In the evaluated CLIP interface comparison, binary feedback with abstention reduces attack success; its positive empirical gap upper bounds leave the true robust gap unresolved.

\paragraph{Contributions.}
\begin{enumerate}[leftmargin=1.4em]
    \item \textbf{Exact minimax limit for target verification.}
    We define the attackability radius
    \[
        \AttackRad_{s,t}:=\inf_{Q\in\Qset_s}\TV(P_t,Q),
    \]
    where $\Qset_s$ is the set of distributions reachable by editing source samples and $\TV$ is total variation. We prove that
    \[
        \sup_{f_t:\X\to[0,1]}
        \left(\E_{P_t}f_t-\sup_{Q\in\Qset_s}\E_Q f_t\right)
        =
        \AttackRad_{s,t}.
    \]
    Thus the optimal passive verifier is governed by a distributional projection of the target onto the attacked source family.

    \item \textbf{Verifier-learnability barrier for public interfaces.}
    We prove that if a deployed verifier is $\varepsilon$-emulable on the attack region and the attacker optimizes the emulator with expected error $\bar\alpha$, then black-box targeted attacks achieve target acceptance within $2\varepsilon+\bar\alpha$ of the white-box optimum.

    \item \textbf{Exact and stable recovery of score-revealing heads.}
    We specialize established score-based recovery ideas to logistic and softmax heads over public features. Full score access reduces the attack to system identification; bounded log-odds error gives a stable emulator when the probe features are well conditioned.

    \item \textbf{Query-cost separation for private low-bandwidth interfaces.}
    We prove a conditional query bound in a stylized hidden-state model, showing how private low-bandwidth feedback can raise attack cost. This does not change the image-only statistical ceiling.

    \item \textbf{Empirical evidence on same-prompt real/diffusion benchmarks.}
    We evaluate public, limited-public, and private interfaces. The evaluated public CLIP verifiers retain negative empirical gap upper bounds after stronger clean or adversarial retraining. We add a finite-state check of the minimax identity, partial fake-to-real transfer to a ResNet-18 victim, and paired image-quality measurements.
\end{enumerate}

Source code: \url{https://github.com/kaikaiyao/pixels-alone-provenance}.

\section{Related Work}

\paragraph{Adversarial robustness.}
Adversarial robustness studies classifiers under worst-case perturbations, robust optimization, transfer, black-box attacks, and gradient-based attacks~\cite{goodfellow2015explaining,papernot2017practical,carlini2017towards,madry2017towards,athalye2018obfuscated}. We use the standard variational characterization of total variation and Sion's minimax theorem to characterize the largest robust target-verification gap achievable by any image-only verifier against a family of edited source distributions. Model extraction~\cite{tramer2016stealing}, knowledge distillation~\cite{hinton2015distilling}, and surrogate attacks~\cite{papernot2017practical} are established tools; our emulator theorem quantifies the effect of score emulation on deployed target acceptance within this source--target analysis.

\paragraph{Passive provenance, detection, and attribution.}
Passive provenance follows the forensic idea that origin traces may remain in the released artifact. Prior work used sensor and camera-model traces~\cite{lukas2006digital,cozzolino2019noiseprint} and learned forensic embeddings for manipulation detection~\cite{cozzolino2018forensictransfer}; other work studies face-manipulation detection~\cite{rossler2019faceforensics,dang2020digital,li2020facexray,shiohara2022sbi}, synthetic-image detection and forensic artifacts~\cite{wang2020cnnspot,ojha2023universal,wang2023dire,cozzolino2024clip,dzanic2020fourier,chandrasegaran2021closer,park2025community,karageorgiou2025spectral,yan2023deepfakebench,yan2024df40,mahara2025methods}, and generator attribution/open-set attribution~\cite{yu2019attributing,yang2023pose,song2024manifpt,girish2021openworld,guarnera2022exploitation,yang2022dna,sun2023cpl,song2025riemannian,nguyen2025forensic}. Prior work shows useful clean accuracy and many empirical failures under attack~\cite{carlini2020deepfake,hussain2021adversarial,wesselkamp2022misleading,yao2025smudged,diao2024vulnerabilities}. Our contribution is an exact source--target minimax characterization and a separate emulator-to-attack bound for a deployed verifier. Neves et al.~\cite{neves2020ganprintr} remove GAN fingerprints to evade fixed face-manipulation detectors; Mavali et al.~\cite{mavali2024fake} study transfer- and query-based attacks on selected AI-image detectors; Wang et al.~\cite{wang2023origin} attribute images through accessible candidate generators and evaluate filter and blur edits. We build on this fixed-system literature to distinguish a source--target limit over all image-only verifiers from the learnability of one deployed interface; the surrogate attack itself is not claimed as new.

\paragraph{Watermarking and active provenance.}
Watermarking modifies the generator, training process, or generated artifact to support later detection or attribution~\cite{yu2020responsible,yu2021artificial,fernandez2023stable,wen2023treering,kim2024wouaf,an2024waves,xu2025invismark}. This is complementary to our setting. We study passive verification, where no trusted generator-side mark or metadata is assumed. For language models, Zhang et al.~\cite{zhang2023watermarks} prove strong-watermark impossibility under quality-oracle and efficiently mixing perturbation-oracle assumptions; for images, Fairoze et al.~\cite{fairoze2025difficulty} study the difficulty of jointly achieving robustness, unforgeability, and public detectability. Our results isolate another mechanism: even if image-level signal exists, the released verifier interface may itself become the attack surface.

\paragraph{Private auditors.}
Keyed passive attribution methods such as SPRINT keep the verifier's reconstruction task private to limit the information available to an adversary~\cite{yao2025authprint}. Our private-interface theorem is not a literal model of any one auditor. It isolates a query-complexity mechanism: hidden verifier state and low-bandwidth feedback can force search instead of direct score optimization.

\section{Target Verification under Adversarial Distribution Shift}
\label{sec:model}

We now define the notation used in the rest of the paper. The image space $\X$ is finite, which matches quantized digital images and keeps the main proof elementary. Appendix~\ref{app:continuous-route} gives an extension to general measurable image spaces under the stated total-variation compactness assumption. Labels are denoted by $y\in\Y$; a label may be \textsc{Human}, an aggregate \textsc{AI} class, or a particular generator. The distribution of images with label $y$ is $P_y$.

For a source--target task, $s$ is the source label and $t$ is the target label. A target verifier is a score function
\[
    f_t:\X\to[0,1],
\]
where $f_t(x)$ is the probability or confidence with which image $x$ is accepted as target $t$. If $d:\X\to[0,1]$ is a binary AI detector, then $f_\AI=d$ and $f_\Human=1-d$. If $g:\X\to\mathcal{P}(\{1,\dots,m\})$ is a multiclass attribution model, then $f_t(x)=g_t(x)$.

The edit correspondence $\Gamma:\X\to2^\X\setminus\{\emptyset\}$ specifies which edited outputs are allowed from each source image. A stochastic editor is a Markov kernel $A(\cdot\mid x)$ supported on $\Gamma(x)$. Applying this editor to $X\sim P_s$ gives an attacked source distribution $A_\#P_s$. We collect all such distributions as
\[
    \Qset_s
    :=
    \{A_\#P_s:\supp(A(\cdot\mid x))\subseteq\Gamma(x)\ \forall x\in\X\}.
\]
The target label changes the objective, but not the reachable family $\Qset_s$. The set $\Qset_s$ is nonempty, convex, and compact; the short proof is in Appendix~\ref{app:proofs-barrier1}.

For a fixed pair $(s,t)$, define clean target acceptance and worst-case targeted acceptance by
\[
    \Acc_t(f_t):=\E_{P_t}[f_t],
    \qquad
    W_{s\to t}(f_t;\Gamma):=\sup_{Q\in\Qset_s}\E_Q[f_t].
\]
The robust target-verification gap is
\[
    \Gap_{s,t}(f_t)
    :=
    \Acc_t(f_t)-W_{s\to t}(f_t;\Gamma).
\]
A positive gap means that true target images receive higher average target score than any edited source distribution. A negative gap means an edited source can receive even higher target score than clean target samples.

For an aggregate AI label, we write $P_\AI=\sum_j\pi_jP_j$, where $P_j$ is the distribution of generator $j$ and the weights $\pi_j$ describe the deployment mixture. This mixture notation is used only when discussing coarse AI-vs-human verification.

\section{Barrier I: A Minimax Statistical Ceiling}
\label{sec:barrier1}

Barrier~I asks the most favorable question for the defender: if we could choose any image-only verifier, what robust gap could we guarantee? Define the attackability radius
\[
    \AttackRad_{s,t}
    :=
    \inf_{Q\in\Qset_s}\TV(P_t,Q),
    \qquad
    \TV(P,Q):=\frac12\sum_{x\in\X}|P(x)-Q(x)|.
\]
This is the distance from the target distribution to the closest attacked source distribution.

\begin{theorem}[Attackability-radius law]
\label{thm:attackability-law}
For every source--target pair $(s,t)$,
\[
    \sup_{f_t:\X\to[0,1]}
    \Gap_{s,t}(f_t)
    =
    \AttackRad_{s,t}.
\]
Equivalently,
\[
    \sup_{f_t:\X\to[0,1]}
    \inf_{Q\in\Qset_s}
    \left(\E_{P_t}f_t-\E_Q f_t\right)
    =
    \inf_{Q\in\Qset_s}\TV(P_t,Q).
\]
\end{theorem}

\paragraph{Why this theorem matters.}
The theorem says that the best possible passive verifier is determined by the geometry of distributions, not by a model class or a training algorithm. If edited source distributions can approach the target distribution, no passive image-only method can keep a large robust gap.

\paragraph{Proof sketch.}
For a fixed attacked distribution $Q$, the best $[0,1]$-valued test separates $P_t$ from $Q$ with value exactly $\TV(P_t,Q)$. This is the standard variational form of total variation. The nontrivial step is that the attacker is not fixed: it may choose any $Q\in\Qset_s$. Since $\Qset_s$ and $[0,1]^\X$ are compact convex sets and the payoff is affine in both arguments, Sion's minimax theorem~\cite{sion1958minimax} lets us exchange the verifier and attacker optimizations. The value becomes the closest total-variation distance from $P_t$ to $\Qset_s$. The full proof, including the compactness and total-variation lemmas, is in Appendix~\ref{app:proofs-barrier1}.

\begin{corollary}[Universal passive ceiling]
\label{cor:universal-ceiling}
For every passive image-only verifier $f_t$,
\[
    \Gap_{s,t}(f_t)\le\AttackRad_{s,t}.
\]
Thus if $\AttackRad_{s,t}=0$, no passive verifier can maintain positive robust target separation under the edit class $\Gamma$.
\end{corollary}

\begin{proposition}[Coarsening can lower the target ceiling]
\label[proposition]{prop:coarsening}
If an aggregate target is $P_\AI=\sum_j\pi_jP_j$, then for any source label $s$,
\[
    \AttackRad_{s,\AI}\le\sum_j\pi_j\AttackRad_{s,j}.
\]
\end{proposition}

\paragraph{Interpretation.}
A coarse label such as \textsc{AI} is a mixture of generator-specific targets, so it can be easier to hit than a particular generator target. This is why binary detection and model attribution need not have the same robustness profile. The proof of Proposition~\ref{prop:coarsening}, and finite examples showing that human-evasion and model-impersonation radii have no universal ordering, are in Appendix~\ref{app:proofs-barrier1}. Appendix~\ref{app:sandwich} explains how this statistical ceiling relates to the empirical gap upper bounds reported later: the image experiments diagnose failure on the evaluated samples without directly estimating \(\AttackRad_{s,t}\).

\section{Barrier II: Learnable Interfaces Collapse to Surrogate Attacks}
\label{sec:barrier2}

Barrier~II studies a different object: a fixed deployed verifier. Even when the statistical ceiling is not tight, a public interface may expose enough information to learn a useful surrogate. Let the attack region be
\[
    \U:=\bigcup_{x\in\X}\Gamma(x).
\]
A surrogate $\widehat f_t:\U\to[0,1]$ is an $\varepsilon$-emulator of $f_t$ on $\U$ if
\[
    \|\widehat f_t-f_t\|_{\infty,\U}
    :=
    \sup_{u\in\U}|\widehat f_t(u)-f_t(u)|
    \le\varepsilon.
\]
For an attack map $\widehat a(x)\in\Gamma(x)$, define its surrogate optimization loss
\[
    \alpha(x)
    :=
    \sup_{z\in\Gamma(x)}\widehat f_t(z)-\widehat f_t(\widehat a(x)),
    \qquad
    \bar\alpha:=\E_{X\sim P_s}[\alpha(X)].
\]
Here $\bar\alpha$ is the average loss of the attack optimizer on the surrogate objective.

\begin{theorem}[Surrogate-attack theorem]
\label{thm:surrogate-attack}
If $\widehat f_t$ is an $\varepsilon$-emulator of $f_t$ on $\U$, then any attack map $\widehat a$ satisfies
\[
    \E_{X\sim P_s}[f_t(\widehat a(X))]
    \ge
    W_{s\to t}(f_t;\Gamma)-2\varepsilon-\bar\alpha.
\]
\end{theorem}

\paragraph{Intuition and proof sketch.}
There are only three losses. The surrogate may score the true white-box maximizer incorrectly by at most $\varepsilon$. The optimizer may lose $\alpha(x)$ on the surrogate. Finally, when the chosen point is evaluated by the deployed verifier, the surrogate may again differ by at most $\varepsilon$. Averaging over source samples gives the bound. The full proof is in Appendix~\ref{app:proofs-barrier2}.

\begin{corollary}[Learned emulators suffice]
\label[corollary]{cor:learned-emulator}
Suppose a query procedure returns a random surrogate $\widehat f_t$ such that
\[
    \Prb[\|\widehat f_t-f_t\|_{\infty,\U}\le\varepsilon]\ge1-\delta.
\]
Then any attack map optimized on the returned surrogate with expected loss $\bar\alpha$ satisfies
\[
\Prb\left[
    \E_{X\sim P_s}[f_t(\widehat a(X))]
    \ge
    W_{s\to t}(f_t;\Gamma)-2\varepsilon-\bar\alpha
\right]
\ge1-\delta.
\]
\end{corollary}

This corollary is the main deployment message of Barrier~II. Exact parameter recovery is not required. Any learning procedure that gives a good emulator on the attack region gives a near-white-box targeted attack.

\begin{proposition}[Score-revealing public heads are identifiable]
\label{prop:score-recovery}
Let
\[
    g_H(j\mid x)
    =
    \frac{\exp(\langle h_j,\phi(x)\rangle)}
    {\sum_{\ell=1}^m \exp(\langle h_\ell,\phi(x)\rangle)}
\]
be a softmax attribution head over a public feature map $\phi:\X\to\R^p$, with gauge $h_m=0$. If the interface returns the full probability vector and the attacker can query $p$ points whose feature matrix is invertible, then the head parameters $H=(h_1,\dots,h_m)$ are exactly recoverable from $p$ queries. The binary logistic case is the $m=2$ specialization. Consequently, surrogate attacks on the recovered head match white-box attacks up to optimization error.
\end{proposition}

\paragraph{Why this proposition is useful.}
A score-revealing softmax head exposes log-odds. With public features, those log-odds are linear equations in the unknown head weights. Once the probe features span $\R^p$, recovering the verifier is just solving linear systems. The detailed derivation is in Appendix~\ref{app:proofs-barrier2}.

\begin{proposition}[Stable recovery under logit error]
\label{prop:stable-recovery}
In the setting of Proposition~\ref{prop:score-recovery}, suppose the recovered log-odds at the $p$ probe points have error at most $\eta$ per coordinate, and let $\Phi$ be the probe feature matrix. Then each recovered class vector satisfies
\[
    \|\widehat h_j-h_j\|_2
    \le
    \|\Phi^{-1}\|_2\sqrt{p}\,\eta.
\]
If $\|\phi(u)\|_2\le R$ on the attack region, the induced logit error per class on the attack region is at most
\[
    B:=R\|\Phi^{-1}\|_2\sqrt{p}\,\eta.
\]
Thus each recovered target score $\widehat f_t$ is an $\varepsilon$-emulator of $f_t$, with $\varepsilon\le B/2$ for softmax, or $\varepsilon\le B/4$ for binary logistic in the reference-class gauge.
\end{proposition}

\paragraph{Interpretation.}
The exact-recovery assumption is not a knife-edge. With bounded log-odds error, recovery degrades with the conditioning of the probe matrix and the feature norm on the attack region. Deriving such a bound from noisy or quantized probabilities additionally requires the probabilities entering the ratios to stay bounded away from zero. Thus finite-precision score access can still yield the emulator required by \cref{thm:surrogate-attack}. A proof is in Appendix~\ref{app:proofs-barrier2}.

The recovery propositions require score-revealing interfaces and public features. They do not claim that every public verifier is exactly recoverable, nor that top-1 labels alone always reveal enough information. They are concrete instantiations of the broader learnability barrier.

\section{Private Interfaces Raise Query Cost, Not the Statistical Ceiling}
\label{sec:private}

Barrier~I applies to private and public image-only verifiers. Privacy cannot create provenance information that is absent from the image. What it can change is how hard it is for the attacker to learn and optimize the verifier. We isolate this effect with a simple hidden-state search bound.

\begin{theorem}[Low-bandwidth hidden-state search]
\label{thm:hidden-state}
Let $\Theta$ be a finite hidden verifier-state space, with $\theta\sim\mathrm{Unif}(\Theta)$. Suppose a $q$-query interface returns one of at most $M$ symbols per query. After observing the transcript, the attacker outputs an action $a$. If every fixed action succeeds on at most a $\rho$ fraction of hidden states,
\[
    \sup_a
    \Pr_{\theta\sim\mathrm{Unif}(\Theta)}[\mathrm{success}(a,\theta)]
    \le \rho,
\]
then every adaptive $q$-query attacker has success probability at most
\[
    M^q\rho.
\]
\end{theorem}

\paragraph{Intuition.}
After $q$ queries, the attacker can see at most $M^q$ transcripts. Each transcript leads to one final action. If any fixed action only works on a $\rho$ fraction of hidden states, then the union over all transcripts covers at most an $M^q\rho$ fraction. The conclusion depends on this fixed-action success bound as well as the interface alphabet; low bandwidth alone does not imply security. The full counting proof is in Appendix~\ref{app:proofs-private}.

\begin{corollary}[Hidden-support threshold verifier]
\label[corollary]{cor:hidden-support}
In the hidden-support model with secret $S\subseteq[d]$, $|S|=k$, candidate budget $b$, and threshold $\tau$, a binary interface satisfies
\[
    \Pr[\mathrm{success}]
    \le
    2^q
    \frac{
        \sum_{i=\tau}^{\min\{k,b\}}
        \binom{b}{i}\binom{d-b}{k-i}
    }{
        \binom{d}{k}
    }.
\]
In particular,
\[
    \Pr[\mathrm{success}]
    \le
    2^q
    \binom{b}{\tau}
    \left(\frac{k}{d}\right)^\tau.
\]
\end{corollary}

The hidden-support model is deliberately simple. It should be read as a clean example of the hidden-state theorem, not as a literal model of a deployed provenance service. With abstention, the interface alphabet may grow, but the accepting region can shrink; this is why abstention can reduce attack success even though it adds another possible output symbol. Details and the sparse-support query lower bound are in Appendix~\ref{app:proofs-private}.

\section{Empirical Evaluation}
\label{sec:exp}

The image experiments keep the source--target task fixed and vary what the verifier reveals: full score vectors, single target scores, quantized scores, top-1 labels, binary decisions, and binary decisions with abstention. They measure attacks on fixed verifiers, while a separate finite-state experiment checks Barrier~I where its minimax value is computable. Appendix~\ref{app:exp-details} reports the original per-direction results and controls; Appendix~\ref{app:additional} gives the attack protocol and additional evaluations. The interfaces are local API-style abstractions; no live commercial service is evaluated.

\paragraph{Benchmarks and verifiers.}
We use two same-prompt benchmarks. The controlled three-class benchmark uses 1{,}200 MS COCO prompts~\cite{lin2014microsoft} to create aligned triplets of real, Stable Diffusion~1.5~\cite{rombach2022high}, and SDXL images, with a prompt-level split of 600/300/300 for train/validation/test and 900 held-out test images. The expanded benchmark keeps the same prompt protocol and adds Stable Diffusion~2.1 and PixArt. Verifiers include frozen CLIP~\cite{radford2021learning} representations with learned linear or MLP heads, and off-the-shelf CLIP-based detector baselines. Attack success rate (ASR) is the fraction of attacked source images accepted as the target under the evaluated interface. The additional ResNet-18 and image-quality evaluations also report induced success among initially correct sources, excluding pre-existing target predictions.

\paragraph{Attack protocol.}
The attacker recovers or fits an emulator from interface replies, runs targeted projected gradient descent (PGD) through the public feature extractor and emulator, and evaluates the final images on the deployed verifier. Pixel attacks use $\ell_\infty$ budgets $2/255$, $4/255$, or $8/255$, 40 steps, step size one quarter of the budget, and five restarts. Appendix~\ref{app:additional-protocol} specifies the interface-dependent losses and transcript counts.

\paragraph{Empirical upper bound on the robust gap.}
For an evaluated attack routine $\widetilde a$, we report
\[
    \widetilde\Gap_{s,t}
    :=
    \Acc_t(f_t)
    -
    \E_{X\sim P_s}[f_t(\widetilde a(X))].
\]
Because $W_{s\to t}(f_t;\Gamma)\ge \E[f_t(\widetilde a(X))]$, we have
\[
    \Gap_{s,t}(f_t)\le \widetilde\Gap_{s,t}.
\]
This inequality holds for population expectations: a negative attack-derived upper bound certifies a negative robust gap for that fixed verifier, whereas a positive bound leaves its sign unresolved. The image tables estimate these expectations by held-out sample averages and retain the notation $\widetilde\Gap$; their signs describe the evaluated samples, without a population confidence guarantee. They do not estimate the minimax radius $\AttackRad_{s,t}$ (Appendix~\ref{app:sandwich}).

\paragraph{Controlled check of Barrier~I.}
On $\X=\{0,1\}^6$, five source--target pairs and Hamming radii $0,1,2,3$ give 20 settings. Independent attack-side transport and verifier-side minimax linear programs agree to a maximum absolute discrepancy of $4.44\times10^{-16}$. Across 20 repeats per setting, plug-in mean absolute error falls from $0.02094$ at $n=250$ to $0.00451$ at $n=5{,}000$ (Table~\ref{tab:finite-state}). This checks the computable finite-state identity; the radius for realistic images remains unestimated.

\subsection{Public score access collapses the controlled verifier}
\label{subsec:exp-public-collapse}

The controlled three-class CLIP-softmax verifier is public, score revealing, and simple enough to instantiate Proposition~\ref{prop:score-recovery}. It is not a vacuous model: on the held-out test set it reaches $0.5167$ accuracy and $0.4802$ macro-F1 on the fine-grained three-way task. The induced fake-vs-real detector reaches fake AUC $0.8872$, TPR@5\%FPR $0.6183$, and EER $0.2192$.

\begin{figure}[t]
\centering
\IfFileExists{fig_public_collapse.png}{%
\includegraphics[width=0.99\linewidth]{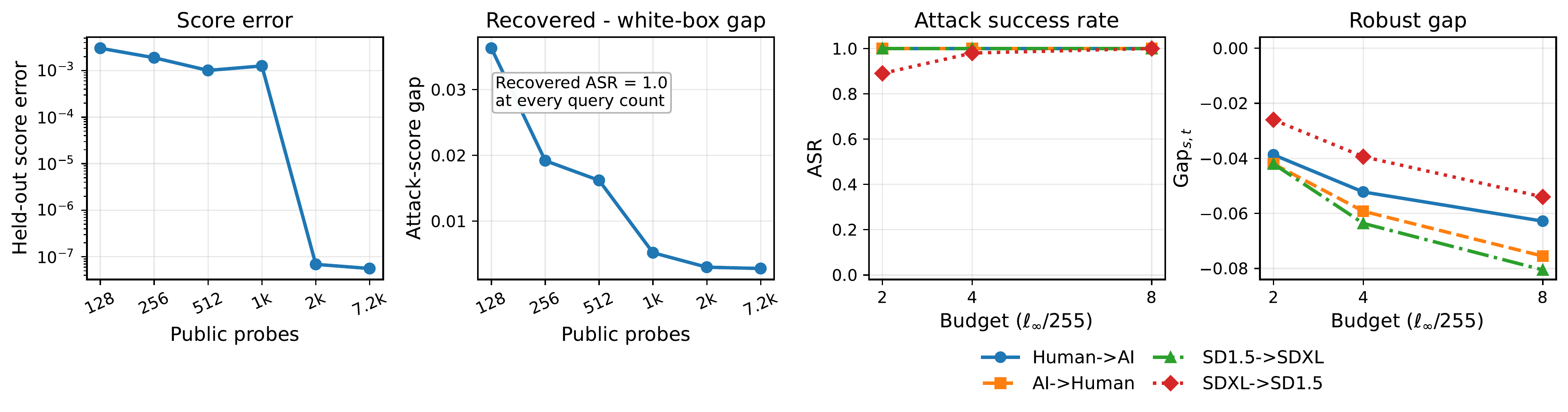}%
}{%
\fbox{\parbox{0.92\linewidth}{\centering Placeholder for \texttt{fig\_public\_collapse.pdf}. The figure shows exact recovery, near-white-box attack performance, collapse across four source--target directions, and negative empirical gap upper bounds.}}%
}
\caption{Controlled public CLIP-softmax verifier on the same-prompt three-class benchmark. Score-revealing queries recover the deployed score function; attacks on the recovered verifier approach white-box performance; all four source--target directions collapse; and every estimated gap upper bound is negative on the evaluated samples.}
\Description{A four-panel figure showing recovery error, recovered-versus-white-box attack gap, attack success rates, and empirical gap upper bounds.}
\label{fig:public-collapse}
\end{figure}

At 7{,}200 probes, the observed score error is $5.60\times10^{-8}$, and attacks on the recovered head differ from direct white-box attacks by $0.0028$ in target score. By 2{,}048 queries, score error is already $6.88\times10^{-8}$ and the recovered-vs-white-box gap is $0.0030$. Attacks optimized on the recovered verifier achieve ASR $1.0$. These measurements instantiate Proposition~\ref{prop:score-recovery}: public score vectors turn black-box attack into head identification followed by white-box optimization. Figure~\ref{fig:public-collapse} reads from left to right as score recovery, difference from direct white-box optimization, directional ASR, and empirical gap upper bounds. At $8/255$, Human$\to$AI has clean-target and attacked-source scores $0.6786$ and $0.7414$ (reported gap $-0.0628$); AI$\to$Human has $0.3338$ and $0.4094$ (reported gap $-0.0755$). Reversing direction changes both distributions and target scores; the two gaps need not coincide. Table~\ref{tab:d1-attack-appendix} gives all directions.

Three source--target directions reach ASR $1.0$ at $2/255$: Human$\to$AI false-flagging, AI$\to$Human evasion, and SD1.5$\to$SDXL impersonation. SDXL$\to$SD1.5 rises from ASR $0.89$ at $2/255$ to $1.0$ by $8/255$. Every empirical upper bound on the robust gap is negative, so edited source images exceed the target label's own clean acceptance, not merely its decision boundary.

\subsection{Interface limits matter more than clean strength}
\label{subsec:exp-interface}

The expanded five-class benchmark tests whether changing the released interface changes attackability. \Cref{tab:interface-main} shows the main pattern. Full vectors and single scores remain highly attackable. Top-1 labels reduce ASR for the linear head but not for the nonlinear head or the off-the-shelf CLIP detector. Plain binary feedback gives ASR $1.0$ for both learned heads, with estimated gap upper bounds $-0.0260$ and $-0.0498$. Binary feedback with abstention reduces ASR and yields positive estimated upper bounds in this CLIP comparison; these positive values do not certify robustness. Appendix~\ref{app:exp-details} reports clean performance, transcript-cap sweeps, and a score-guided random-search control.

\begin{table}[t]
\centering
\caption{Representative $8/255$ interface slice on the expanded five-class benchmark. Entries are ASR; a dash denotes an unevaluated combination. Plain binary feedback remains attackable for both learned heads. Binary feedback with abstention uses the configured per-split cap $7{,}200$ and reduces ASR across these three systems.}
\label{tab:interface-main}
\normalsize
\setlength{\tabcolsep}{6pt}
\begin{tabular}{@{}lccccc@{}}
  \toprule
  \textbf{Verifier}
    & \shortstack[c]{Full\\vector}
    & \shortstack[c]{Single\\score}
    & \shortstack[c]{Top-1\\label}
    & \shortstack[c]{Binary\\verdict}
    & \shortstack[c]{Binary +\\abstain} \\
  \midrule
  Linear softmax CLIP head & 1.00 & 0.91 & 0.75 & 1.00 & 0.00 \\
  Two-layer MLP CLIP head & 1.00 & 0.93 & 1.00 & 1.00 & 0.00 \\
  Off-the-shelf CLIP detector & 1.00 & 1.00 & 1.00 & -- & 0.09 \\
  \bottomrule
\end{tabular}
\end{table}

For the shared-representation CLIP heads, the validation-set $95$th-percentile emulator diagnostic $\epsilon_{95}$ rises from about $0.005$ under full vectors to $0.31$--$0.33$ under a single score, about $0.70$ under plain binary feedback, and $1.0$ under top-1 labels or abstention. This diagnostic is not the uniform attack-region error $\varepsilon$ in \cref{thm:surrogate-attack}; the sweep also does not estimate $\bar\alpha$, so it does not numerically validate the $2\varepsilon+\bar\alpha$ bound. High validation error alone does not imply safety: top-1 attacks still succeed against the MLP and off-the-shelf detector.

\subsection{Stronger retraining still does not restore robustness}
\label{subsec:exp-retrain-control}

A natural objection is that public-interface collapse is caused by weak clean training. \Cref{tab:advd1-summary} tests this directly on the original three-class benchmark. The stronger clean retrain improves accuracy from $0.517$ to $0.869$ and fake AUC from $0.887$ to $0.982$, but its empirical gap upper bounds remain negative. Adversarial retraining reduces the severity of some attacks, but the estimated gap upper bounds remain negative at $8/255$. Supporting per-direction results are in \cref{tab:advd1-attack-appendix}.

\begin{table}[t]
\centering
\caption{Released and retrained models on the original three-class benchmark. Clean columns report test accuracy and fake AUC; attack columns summarize the minimum ASR, maximum (best) gap, or minimum (worst) gap across the four evaluated source--target directions at the indicated budget.}
\label{tab:advd1-summary}
\normalsize
\setlength{\tabcolsep}{4.8pt}
\begin{tabular}{@{}lcccccc@{}}
\toprule
\textbf{Model} & \textbf{Acc.} & \textbf{AUC} & \shortstack[c]{Min ASR\\$2/255$} & \shortstack[c]{Min ASR\\$8/255$} & \shortstack[c]{Best gap\\$2/255$} & \shortstack[c]{Worst gap\\$8/255$} \\
\midrule
Released & 0.517 & 0.887 & 0.890 & 1.000 & $-0.026$ & $-0.081$ \\
Clean retrain & 0.869 & 0.982 & 0.967 & 1.000 & $-0.169$ & $-0.364$ \\
Adv. retrain & 0.821 & 0.947 & 0.900 & 1.000 & $-0.024$ & $-0.185$ \\
\bottomrule
\end{tabular}
\end{table}

\subsection{Transfer beyond CLIP and paired image quality}
\label{subsec:exp-additional}

An RGB ResNet-18~\cite{he2016deep} victim reaches test fake AUC $0.9862$. At $8/255$, attacks through a CLIP-feature emulator induce fake-to-real success on $45/117$ initially correct sources ($38.46\%$), but real-to-fake success on $0/122$. Both estimated gap upper bounds remain positive (Table~\ref{tab:resnet-transfer}). This supports partial directional transfer across representations for one victim and setting, rather than general collapse beyond CLIP.

In a paired evaluation of four attack directions, each budget has 64 attack instances, of which 50 start from correct predictions. At $2/255$, $45/50$ are induced successes, with pooled success-only medians of $0.9784$ structural similarity (SSIM)~\cite{wang2004ssim}, $0.0063$ learned perceptual image patch similarity (LPIPS-Alex)~\cite{zhang2018unreasonable}, and $43.82$~dB peak signal-to-noise ratio (PSNR). At $8/255$, induced success reaches $50/50$, while per-direction median SSIM falls to $0.8011$--$0.8534$ and LPIPS-Alex rises to $0.0770$--$0.1365$ (Table~\ref{tab:paired-quality}). The larger budget trades paired image quality for five additional successes. Figure~\ref{fig:attack-examples} shows a success and a failure; Appendix~\ref{app:additional-quality} includes paired examples at both budgets. These examples and metrics do not constitute a human perceptual study.

\begin{figure}[t]
\centering
\includegraphics[width=\linewidth]{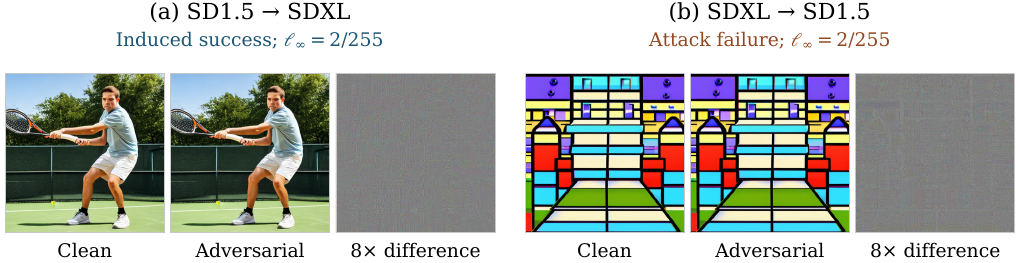}
\caption{Paired examples at an $\ell_\infty$ budget of $2/255$. Each triplet shows the clean image, the attacked image, and the signed perturbation amplified eightfold around mid-gray. The SD1.5$\to$SDXL example induces target acceptance, while the SDXL$\to$SD1.5 example fails. These selected examples illustrate both outcomes.}
\label{fig:attack-examples}
\end{figure}

\section{Discussion and Limitations}
\label{sec:discussion}

The paper separates two failure modes that are easy to confuse. The attackability-radius law is an image-only statistical ceiling: it applies even to the best possible passive verifier and depends only on $P_t$, $P_s$, and $\Gamma$. The surrogate-attack theorem is an interface statement: it explains how one deployed public verifier can fail before that ceiling is reached.

The experiments support the second mechanism in the evaluated settings. Public CLIP-based interfaces are learnable enough that targeted attacks reach or exceed the target label's own clean acceptance. Stronger clean retraining and adversarial retraining improve clean metrics but do not restore positive empirical upper bounds on the robust gap.

\paragraph{What Barrier~I does and does not say.}
Barrier~I is an optimistic benchmark for passive verification. It asks what would be possible if the defender could choose the best possible image-only verifier, ignoring architecture, computation, and training data. A small $\AttackRad_{s,t}$ means the edited source family can get close to the target distribution, so there is little robust separation available to any passive method. A large $\AttackRad_{s,t}$ does not mean a deployed model is safe. It only means that the image distributions leave room for a good verifier in principle.

\paragraph{What Barrier~II adds.}
Barrier~II explains how a real verifier can lose that room. The theorem does not require the attacker to match the target distribution. It only requires the attacker to learn enough of the deployed score on the attack region. This is why the empirical upper bound $\widetilde\Gap$ is useful: when it is negative, the attacked source samples already receive higher target score than clean target samples under the deployed verifier. That is a failure of the deployed verifier, even if the unknown statistical ceiling $\AttackRad_{s,t}$ is not directly estimated.

\paragraph{What the private-interface results mean.}
The private-interface theorem is conditional on hidden verifier state and the fixed-action success bound $\rho$. It shows how limited transcripts can constrain search under those assumptions. In the CLIP experiments, abstention changes the accepting region and reduces measured attack success; its positive estimated gap upper bounds leave robustness unresolved. Such interface restrictions cannot overcome a small attackability radius.

\paragraph{How to use the results.}
For a new passive provenance system, clean accuracy should be treated as only a first diagnostic. A more relevant evaluation should report source--target attack directions, the interface exposed to the attacker, attack success rate, and the empirical upper bound $\widetilde\Gap$ on the robust gap. If public scores are released, the system should also be tested for surrogate learnability. If a private interface is used, the evaluation should test whether its accepting region can still be discovered through queries. The appendices give the full proofs and the extra empirical tables needed to audit these claims.

\paragraph{What should not be inferred.}
Our results do not say that passive provenance is useless. They say that a passive verifier has two separate obligations. First, the source and target distributions must remain separated after the allowed edits; this is the Barrier~I question. Second, the deployed interface must not hand the attacker a useful optimization signal; this is the Barrier~II question. A system can have strong clean accuracy and still fail the second obligation. Conversely, a private interface can make attacks harder while still leaving the first obligation unresolved.

Several limitations remain. First, the finite-state calculation does not estimate $\AttackRad_{s,t}$ for realistic image distributions; Appendix~\ref{app:sandwich} describes bounds whose tightness remains unresolved. Second, the private-interface theorem is a stylized query-complexity model, not a security proof for a deployed auditor. Third, the image evidence is primarily CLIP-centered; the additional ResNet-18 result covers one victim, one seed, one score interface, and partial transfer in one direction. All image experiments use same-prompt real/diffusion benchmarks and bounded pixel edits. Semantic edits, non-diffusion sources, and live commercial APIs are not evaluated. Finally, trusted marks, keys, signatures, or authenticated metadata require a separate provenance threat model; no active or hybrid system is evaluated here.

The deployment lesson is that clean provenance accuracy is not a robustness certificate. A passive verifier must be evaluated jointly with the interface through which it is exposed. In practice, this means that a robust evaluation should state three objects together: the source--target direction, the edit class, and the released interface. Omitting any one of these can invert the conclusion. A detector may look strong when averaged over a binary human-versus-AI task, weak when the same images are evaluated as generator impersonation, and weaker still when the public interface reveals scores that make the verifier easy to emulate. In the evaluated CLIP systems, abstention improves the attack-derived gap estimate without changing the source distributions. This observation does not establish the sign of the true robust gap or a change in the statistical ceiling.

\section{Conclusion}

We formulated passive image provenance as source--target verification under adversarial distribution shift. The resulting minimax game has an exact value: the best robust gap achievable by any image-only verifier is the total-variation distance from the target distribution to the attacked source family. We then showed a separate public-interface barrier: once a verifier is learnable on the attack region, surrogate black-box attacks approach white-box performance. A finite-state calculation checks the minimax identity. On the same-prompt image benchmarks, the evaluated public CLIP verifiers fail under targeted pixel attacks; a ResNet-18 victim shows partial fake-to-real transfer, and paired metrics quantify the image-quality trade-off. Abstention reduces measured attack success in the CLIP comparison, while positive estimated gap upper bounds leave robustness unresolved. Robust passive provenance therefore requires both statistical source--target analysis and interface-aware evaluation. The central recommendation is to report these two axes separately: estimate or bracket the image-level source--target separation where possible, and stress-test the exact interface that users or adversaries will observe. This framing also makes negative results actionable: a failure caused by a small distributional radius calls for stronger provenance signals or active marks, whereas a failure caused by a learnable interface calls for interface redesign and a renewed black-box audit.

\begin{ack}
We thank our anonymous reviewers for their valuable feedback, which substantially improved the paper. This work was supported by the Edinburgh International Data Facility (EIDF) and the Data-Driven Innovation Programme at the University of Edinburgh.
\end{ack}

\appendix

\section{Additional Formal Details for Barrier I}
\label{app:proofs-barrier1}

\subsection{Reachable attacked distributions}

\begin{lemma}[Reachable sets are convex and compact]
\label[lemma]{lem:qset-convex}
For every source label $s$, the set $\Qset_s$ is nonempty, convex, and compact.
\end{lemma}

\begin{proof}
Let
\[
\Kset_\Gamma:=
\left\{
A(\cdot\mid x)_{x\in\X}:
\begin{array}{l}
A(\cdot\mid x)\text{ is a distribution}\\
\text{supported on }\Gamma(x)\text{ for every }x\in\X
\end{array}
\right\}.
\]
Because $\X$ is finite and each $\Gamma(x)$ is nonempty, $\Kset_\Gamma$ is a nonempty product of simplices, hence compact and convex. The map
\[
\Psi:\Kset_\Gamma\to\mathcal{P}(\X),\qquad \Psi(A)=A_\#P_s,
\]
is linear and continuous. Therefore $\Qset_s=\Psi(\Kset_\Gamma)$ is nonempty, compact, and convex.
\end{proof}

\begin{proposition}[White-box optimum is pointwise]
\label{prop:white-box-pointwise}
For every target verifier $f_t:\X\to[0,1]$,
\[
    W_{s\to t}(f_t;\Gamma)=\E_{X\sim P_s}\Big[\sup_{z\in\Gamma(X)} f_t(z)\Big].
\]
In particular, randomized attacks do not outperform a pointwise maximizer.
\end{proposition}

\begin{proof}
For any admissible attack kernel $A$ and each $x\in\X$,
\[
    \E_{Z\sim A(\cdot\mid x)}[f_t(Z)]\le \sup_{z\in\Gamma(x)}f_t(z).
\]
Averaging over $X\sim P_s$ gives the upper bound. Since $\X$ is finite and $\Gamma(x)$ is nonempty, for each $x$ there exists a maximizer $a_t(x)\in\argmax_{z\in\Gamma(x)}f_t(z)$. The deterministic kernel $A_t(\cdot\mid x)=\delta_{a_t(x)}$ attains equality.
\end{proof}

\subsection{Total variation preliminaries}

\begin{lemma}[Variational form of total variation]
\label[lemma]{lem:tv-variational}
For any two distributions $P,Q$ on $\X$,
\[
    \TV(P,Q)=\sup_{f:\X\to[0,1]}\bigl(\E_P[f]-\E_Q[f]\bigr).
\]
\end{lemma}

\begin{proof}
For any event $B\subseteq\X$, the indicator $\one_B$ belongs to $[0,1]^\X$, so
\[
\sup_{f:\X\to[0,1]} (\E_P[f]-\E_Q[f])\ge \sup_{B\subseteq\X}(P(B)-Q(B))=\TV(P,Q).
\]
Conversely, for any $f:\X\to[0,1]$,
\[
\E_P[f]-\E_Q[f]
=\sum_{x\in\X} f(x)(P(x)-Q(x))
\le \sum_{x:P(x)\ge Q(x)}(P(x)-Q(x))
=\TV(P,Q).
\]
Taking the supremum proves equality.
\end{proof}

\begin{proposition}[Best verifier against a fixed attacked distribution]
\label{prop:fixedQ-best-test}
Fix a target distribution $P_t$ and any comparison distribution $Q$ on $\X$. Define
\[
A_Q:=\{x\in\X: P_t(x)\ge Q(x)\},
\qquad
f_Q:=\one_{A_Q}.
\]
Then
\[
\E_{P_t}[f_Q]-\E_Q[f_Q]=\TV(P_t,Q).
\]
\end{proposition}

\begin{proof}
By definition of $A_Q$,
\[
\E_{P_t}[f_Q]-\E_Q[f_Q]
=P_t(A_Q)-Q(A_Q)
=\sum_{x\in A_Q}(P_t(x)-Q(x)).
\]
Because $P_t(x)-Q(x)\ge0$ exactly on $A_Q$ and the total positive mass equals the total variation distance, the claim follows.
\end{proof}

\subsection{Full proof of the attackability-radius law}

\begin{proof}[Proof of \cref{thm:attackability-law}]
For every verifier $f_t$,
\[
\begin{aligned}
\Gap_{s,t}(f_t)
&=\Acc_t(f_t)-W_{s\to t}(f_t;\Gamma)\\
&=\E_{P_t}[f_t]-\sup_{Q\in\Qset_s}\E_Q[f_t]\\
&=\inf_{Q\in\Qset_s}\left(\E_{P_t}[f_t]-\E_Q[f_t]\right).
\end{aligned}
\]
Let
\[
L(f_t,Q):=\E_{P_t}[f_t]-\E_Q[f_t].
\]
By \cref{lem:qset-convex}, $\Qset_s$ is nonempty, convex, and compact. Since $\X$ is finite, $[0,1]^\X$ is also convex and compact. The payoff $L$ is affine in $f_t$ and affine in $Q$. Sion's minimax theorem therefore applies:
\[
    \sup_{f_t}\inf_{Q\in\Qset_s}L(f_t,Q)
    =
    \inf_{Q\in\Qset_s}\sup_{f_t}L(f_t,Q).
\]
For any fixed attacked distribution $Q$, \cref{lem:tv-variational} gives
\[
    \sup_{f_t:\X\to[0,1]}L(f_t,Q)=\TV(P_t,Q).
\]
Combining these identities yields
\[
    \sup_{f_t:\X\to[0,1]}\Gap_{s,t}(f_t)
    =
    \inf_{Q\in\Qset_s}\TV(P_t,Q)
    =
    \AttackRad_{s,t}.
\]
\end{proof}

\subsection{Coarsening and directional asymmetry}

\begin{proof}[Proof of Proposition~\ref{prop:coarsening}]
By compactness of $\Qset_s$, for each $j$ there exists
\[
Q_j^\star\in\argmin_{Q\in\Qset_s}\TV(P_j,Q)
\]
with $\TV(P_j,Q_j^\star)=\AttackRad_{s,j}$. Define $\overline Q:=\sum_j\pi_jQ_j^\star$. By convexity, $\overline Q\in\Qset_s$. Joint convexity of total variation gives
\[
\TV(P_\AI,\overline Q)
=
\TV\Big(\sum_j\pi_jP_j,\sum_j\pi_jQ_j^\star\Big)
\le
\sum_j\pi_j\TV(P_j,Q_j^\star)
=
\sum_j\pi_j\AttackRad_{s,j}.
\]
Since $\AttackRad_{s,\AI}=\inf_{Q\in\Qset_s}\TV(P_\AI,Q)\le\TV(P_\AI,\overline Q)$, the result follows.
\end{proof}

\begin{proposition}[No universal ordering between human evasion and model impersonation]
\label{prop:no-ordering}
There is no universal inequality relating $\AttackRad_{s,\Human}$ and $\AttackRad_{s,t}$ for a specific model target $t$. Both strict inequalities can occur.
\end{proposition}

\begin{proof}
We give two finite-space examples. Let $\X=\{a,b,c\}$, $P_s=\delta_a$, and $\Gamma(a)=\{a,b\}$.

First let $P_\Human=\tfrac12\delta_b+\tfrac12\delta_c$ and $P_t=\delta_c$. Since every reachable attacked distribution is supported on $\{a,b\}$, choosing $Q=\delta_b$ gives $\AttackRad_{s,\Human}=1/2$, while every reachable $Q$ is disjoint from $\delta_c$, so $\AttackRad_{s,t}=1$.

Second let $P_\Human=\delta_c$ and $P_t=\delta_b$. Then choosing $Q=\delta_b$ gives $\AttackRad_{s,t}=0$, while every reachable $Q$ is disjoint from $\delta_c$, so $\AttackRad_{s,\Human}=1$.
\end{proof}

\section{Additional Formal Details for Barrier II}
\label{app:proofs-barrier2}

\subsection{Proof of the surrogate-attack theorem}

\begin{proof}[Proof of \cref{thm:surrogate-attack}]
Fix $x\in\X$ and let
\[
z_x^\star\in\argmax_{z\in\Gamma(x)}f_t(z).
\]
Since $\widehat f_t$ is an $\varepsilon$-emulator on $\U$ and $z_x^\star\in\Gamma(x)\subseteq\U$,
\[
\widehat f_t(z_x^\star)
\ge f_t(z_x^\star)-\varepsilon
=
\sup_{z\in\Gamma(x)}f_t(z)-\varepsilon.
\]
By definition of $\alpha(x)$,
\[
\widehat f_t(\widehat a(x))
\ge
\sup_{z\in\Gamma(x)}\widehat f_t(z)-\alpha(x)
\ge
\widehat f_t(z_x^\star)-\alpha(x).
\]
Applying the emulator bound again at $\widehat a(x)$ gives
\[
\begin{aligned}
f_t(\widehat a(x))
&\ge \widehat f_t(\widehat a(x))-\varepsilon\\
&\ge \sup_{z\in\Gamma(x)}f_t(z)-2\varepsilon-\alpha(x).
\end{aligned}
\]
Averaging over $X\sim P_s$ and using Proposition~\ref{prop:white-box-pointwise},
\[
\begin{aligned}
\E_{X\sim P_s}[f_t(\widehat a(X))]
&\ge
\E_{X\sim P_s}\left[\sup_{z\in\Gamma(X)}f_t(z)\right]
-2\varepsilon-\bar\alpha\\
&=W_{s\to t}(f_t;\Gamma)-2\varepsilon-\bar\alpha.
\end{aligned}
\]
\end{proof}

\begin{proof}[Proof of \cref{cor:learned-emulator}]
On the event $\{\|\widehat f_t-f_t\|_{\infty,\U}\le\varepsilon\}$, \cref{thm:surrogate-attack} applies. This event has probability at least $1-\delta$ by assumption.
\end{proof}

\subsection{Exact recovery of score-revealing heads}

\begin{proof}[Proof of Proposition~\ref{prop:score-recovery}]
For the softmax head, fix the gauge $h_m=0$. For each probe point $x_i$ and each class $j<m$,
\[
\log\frac{g_H(j\mid x_i)}{g_H(m\mid x_i)}
=
\langle h_j,\phi(x_i)\rangle-\langle h_m,\phi(x_i)\rangle
=
\langle h_j,\phi(x_i)\rangle.
\]
Stacking the $p$ equations for a fixed class $j$ gives
\[
\Phi h_j=
\begin{bmatrix}
\log \frac{g_H(j\mid x_1)}{g_H(m\mid x_1)}\\
\vdots\\
\log \frac{g_H(j\mid x_p)}{g_H(m\mid x_p)}
\end{bmatrix}.
\]
Since $\Phi$ is invertible, each $h_j$ is exactly determined. The gauge determines $h_m=0$, so the whole head $H$ is recovered. The binary logistic case is obtained by taking $m=2$ and converting returned probabilities into logits.

Exact recovery gives an exact emulator of every target coordinate $f_t(x)=g_H(t\mid x)$. Applying \cref{thm:surrogate-attack} with $\varepsilon=0$ shows that attacks on the recovered head match white-box attacks up to optimization error.
\end{proof}

\begin{proof}[Proof of Proposition~\ref{prop:stable-recovery}]
For each class $j<m$, let $e_j\in\R^p$ be the vector of log-odds errors at the probe points. By assumption, $\|e_j\|_2\le\sqrt{p}\,\eta$. The recovered vector satisfies
\[
\Phi\widehat h_j=\Phi h_j+e_j,
\]
so
\[
\widehat h_j-h_j=\Phi^{-1}e_j.
\]
Therefore
\[
\|\widehat h_j-h_j\|_2
\le
\|\Phi^{-1}\|_2\|e_j\|_2
\le
\|\Phi^{-1}\|_2\sqrt{p}\,\eta.
\]
For any $u\in\U$ with $\|\phi(u)\|_2\le R$,
\[
|\langle \widehat h_j-h_j,\phi(u)\rangle|
\le
\|\widehat h_j-h_j\|_2\|\phi(u)\|_2
\le
R\|\Phi^{-1}\|_2\sqrt{p}\,\eta.
\]
Let $\ell(u)$ and $\widehat\ell(u)$ denote the true and recovered logit vectors in the same gauge, with $\ell_m(u)=\widehat\ell_m(u)=0$. The preceding per-class bound gives $\|\widehat\ell(u)-\ell(u)\|_\infty\le B$ uniformly on $\U$. For a softmax target coordinate $s_t(\ell)$,
\[
\|\nabla s_t(\ell)\|_1=2s_t(\ell)(1-s_t(\ell))\le\tfrac12.
\]
The mean-value bound therefore gives $\|\widehat f_t-f_t\|_{\infty,\U}\le B/2$. In the binary reference-class gauge, the target probability is a logistic function of one log-odds value; its derivative is at most $1/4$, giving the sharper bound $B/4$.
\end{proof}

\section{Private-Interface Counting Proofs}
\label{app:proofs-private}

\begin{proof}[Proof of \cref{thm:hidden-state}]
Condition on the attacker's internal randomness, yielding a deterministic adaptive strategy. There are at most $M^q$ possible transcripts. Each transcript determines one final action $a$. By assumption, the set of hidden states on which any fixed action succeeds has size at most $\rho|\Theta|$. If the attacker succeeds, the realized hidden state must lie in the success set of the action corresponding to the realized transcript. Thus the success set is contained in a union of at most $M^q$ sets, each of size at most $\rho|\Theta|$. The success probability is therefore at most $M^q\rho$. Averaging over the attacker's randomness preserves the bound.
\end{proof}

\subsection{Hidden-support instantiation}

Fix integers $d,k,b,\tau$ with $1\le\tau\le b\le d$ and $\tau\le k\le d$. A verifier secretly samples a support $S\subseteq[d]$, $|S|=k$, uniformly from all $k$-subsets. A candidate forged sample is represented by a set $B\subseteq[d]$ of activated coordinates with $|B|\le b$. The verifier accepts if and only if $|B\cap S|\ge\tau$.

For a fixed set $B$ with $|B|=b$, define
\[
N_{d,k,b,\tau}:=
\sum_{i=\tau}^{\min\{k,b\}}\binom{b}{i}\binom{d-b}{k-i}.
\]
Then
\[
\Pr[|S\cap B|\ge\tau]
=
\frac{N_{d,k,b,\tau}}{\binom{d}{k}},
\]
because achieving intersection size $i$ requires choosing $i$ elements from $B$ and $k-i$ from its complement. Applying \cref{thm:hidden-state} with $M=2$ and
\[
\rho=\frac{N_{d,k,b,\tau}}{\binom{d}{k}}
\]
gives \cref{cor:hidden-support}. If $|B|<b$, enlarge it to any superset of size $b$, which can only increase its acceptance probability.

For the simpler upper bound, the event $|S\cap B|\ge\tau$ implies that some $\tau$-subset $T\subseteq B$ lies inside $S$. By a union bound,
\[
\Pr[|S\cap B|\ge\tau]
\le
\sum_{\substack{T\subseteq B\\ |T|=\tau}}\Pr[T\subseteq S].
\]
For each fixed $T$,
\[
\Pr[T\subseteq S]
=
\frac{\binom{d-\tau}{k-\tau}}{\binom{d}{k}}
=\prod_{r=0}^{\tau-1}\frac{k-r}{d-r}
\le
\left(\frac{k}{d}\right)^\tau.
\]
Hence
\[
\Pr[|S\cap B|\ge\tau]
\le
\binom{b}{\tau}\left(\frac{k}{d}\right)^\tau.
\]
Multiplying by $2^q$ gives the $q$-query bound.

\begin{corollary}[Explicit sparse-support query lower bound]
\label{cor:sparse-query-lb}
Assume $ebk/(\tau d)<1$. Any binary-query attacker that succeeds with probability at least $\eta$ in the hidden-support threshold model must make at least
\[
q\ge \tau\log_2\left(\frac{\tau d}{ebk}\right)-\log_2\left(\frac{1}{\eta}\right)
\]
queries.
\end{corollary}

\begin{proof}
If success probability is at least $\eta$, then the preceding bound gives
\[
\eta\le 2^q\left(\frac{ebk}{\tau d}\right)^\tau,
\]
using $\binom{b}{\tau}\le(eb/\tau)^\tau$. Rearranging proves the claim.
\end{proof}

\section{Additional Empirical Results}
\label{app:exp-details}

This appendix is the empirical audit trail for \cref{sec:exp}. It is organized around the three empirical claims made in the main text rather than as a collection of independent tables. First, \cref{subsec:app-controlled-public} verifies the public-head recovery mechanism used in \cref{subsec:exp-public-collapse}: score access lets the attacker identify the deployed head and then attack it nearly as a white-box model. Second, \cref{subsec:app-expanded-interface} supports \cref{subsec:exp-interface}: changing the released interface changes attackability more than changing clean accuracy alone. Third, \cref{subsec:app-retrain-controls} supports \cref{subsec:exp-retrain-control}: stronger clean training and adversarial retraining retain negative estimated gap upper bounds in the evaluated public setting.

Unless stated otherwise, interface results are reported at the representative $8/255$ slice used in the main text, and $\widetilde\Gap$ is the empirical upper bound defined in \cref{sec:exp}. A table is included only when it checks a specific link in the main argument: recovery quality, directional collapse, clean operating point, interface bandwidth, transcript caps, score-guided random search, or retraining controls.

\subsection{Controlled public CLIP-softmax baseline}
\label{subsec:app-controlled-public}

The controlled three-class experiment in \cref{subsec:exp-public-collapse} uses a public CLIP-softmax verifier. The purpose of \cref{tab:d1-recovery-appendix} is to connect Proposition~\ref{prop:score-recovery} to the measured attack pipeline. The first numerical column tracks score emulation error, the second tracks the gap between attacks optimized on the recovered model and direct white-box attacks, and the final column checks whether the recovered model is already sufficient for targeted success. Thus the table is not merely a query sweep: it shows the transition from approximate recovery to operationally white-box behavior.

\begin{table}[H]
\centering
\caption{Exact-recovery details for the controlled public CLIP-softmax baseline on the three-class benchmark. The table supports \cref{subsec:exp-public-collapse} by showing that score emulation becomes operationally white-box by 2{,}048 queries.}
\label{tab:d1-recovery-appendix}
\small
\begin{tabular}{@{}rrrr@{}}
\toprule
\textbf{Queries} & \textbf{Score error} & \textbf{Rec.-WB gap} & \textbf{Recovered ASR} \\
\midrule
128 & $3.03\times 10^{-3}$ & 0.0363 & 1.000 \\
256 & $1.89\times 10^{-3}$ & 0.0192 & 1.000 \\
512 & $1.01\times 10^{-3}$ & 0.0162 & 1.000 \\
1{,}024 & $1.26\times 10^{-3}$ & 0.0052 & 1.000 \\
2{,}048 & $6.88\times 10^{-8}$ & 0.0030 & 1.000 \\
7{,}200 & $5.60\times 10^{-8}$ & 0.0028 & 1.000 \\
\bottomrule
\end{tabular}
\end{table}

Once the public head is recovered, the next question is whether collapse is confined to one source--target direction. \Cref{tab:d1-attack-appendix} answers that question. It expands the aggregate statement in \cref{fig:public-collapse}: human false-flagging, AI evasion, and generator impersonation all reach ASR $1.0$ by the largest evaluated budget, and every reported $\widetilde\Gap$ is negative. The last two columns make the negative-gap interpretation explicit on the held-out samples: attacked source images receive higher average target score than the clean target images themselves.

\begin{table}[H]
\centering
\caption{Per-direction targeted attacks for the original controlled public CLIP-softmax verifier on the three-class benchmark. This table supports the negative-gap claim in \cref{subsec:exp-public-collapse}: all four source--target directions reach ASR $1.0$ by $8/255$, and every empirical gap upper bound is negative.}
\label{tab:d1-attack-appendix}
\normalsize
\setlength{\tabcolsep}{5pt}
\begin{tabular}{@{}llrrrr@{}}
\toprule
\textbf{Direction} & \textbf{Budget} & \textbf{ASR} & $\boldsymbol{\widetilde\Gap}$ & \shortstack[c]{\textbf{Clean target}\\\textbf{acc.}} & \shortstack[c]{\textbf{Adv. target}\\\textbf{score}} \\
\midrule
Human$\to$AI & $2/255$ & 1.0000 & $-0.0387$ & 0.6786 & 0.7173 \\
Human$\to$AI & $4/255$ & 1.0000 & $-0.0522$ & 0.6786 & 0.7308 \\
Human$\to$AI & $8/255$ & 1.0000 & $-0.0628$ & 0.6786 & 0.7414 \\
AI$\to$Human & $2/255$ & 1.0000 & $-0.0420$ & 0.3338 & 0.3758 \\
AI$\to$Human & $4/255$ & 1.0000 & $-0.0592$ & 0.3338 & 0.3930 \\
AI$\to$Human & $8/255$ & 1.0000 & $-0.0755$ & 0.3338 & 0.4094 \\
SD1.5$\to$SDXL & $2/255$ & 1.0000 & $-0.0420$ & 0.3590 & 0.4010 \\
SD1.5$\to$SDXL & $4/255$ & 1.0000 & $-0.0636$ & 0.3590 & 0.4226 \\
SD1.5$\to$SDXL & $8/255$ & 1.0000 & $-0.0805$ & 0.3590 & 0.4395 \\
SDXL$\to$SD1.5 & $2/255$ & 0.8900 & $-0.0260$ & 0.3345 & 0.3605 \\
SDXL$\to$SD1.5 & $4/255$ & 0.9800 & $-0.0394$ & 0.3345 & 0.3739 \\
SDXL$\to$SD1.5 & $8/255$ & 1.0000 & $-0.0540$ & 0.3345 & 0.3885 \\
\bottomrule
\end{tabular}
\end{table}

The final check for the controlled public model is whether coarsening to a binary AI target avoids the attribution-style impersonation failure discussed after Proposition~\ref{prop:coarsening}. It does not in this setting. At $8/255$, all evaluated human$\to$AI mixture attacks---both oracle-mixture and coarse-only, with mixture weight $\pi\in\{0,0.25,0.5,0.75,1\}$---finish with ASR $1.0$. This observation is used only as an operational check; the formal coarsening statement remains the distributional inequality in \cref{prop:coarsening}.

\subsection{Expanded five-class interface study}
\label{subsec:app-expanded-interface}

The expanded benchmark adds more generators and interface variants. \Cref{tab:clean-matrix-appendix} first establishes the clean operating point. Its role is to prevent a misleading reading of \cref{tab:interface-main}: the selected linear, nonlinear, and off-the-shelf CLIP systems are not arbitrary weak detectors, while CNNSpot and DIRE are retained only as weak negative controls on this same-prompt benchmark. The main interface conclusions are therefore drawn from the three representative systems used in \cref{tab:interface-main}, not from the weak controls.

\begin{table}[H]
\centering
\caption{Clean matrix on the expanded five-class benchmark. This table supports the setup of \cref{subsec:exp-interface}: the three representative CLIP-based systems have meaningful clean performance, while CNNSpot and DIRE are included as weak off-the-shelf controls rather than as main robustness claims.}
\label{tab:clean-matrix-appendix}
\small
\setlength{\tabcolsep}{3pt}
\resizebox{\linewidth}{!}{%
\begin{tabular}{@{}lccccc l@{}}
\toprule
\textbf{Verifier} & \textbf{Acc.} & \textbf{Macro-F1} & \textbf{Fake AUC} & \textbf{TPR@5\%FPR} & \textbf{EER} & \textbf{Role} \\
\midrule
Binary logistic on CLIP & 0.8007 & 0.4479 & 0.8470 & 0.4033 & 0.2288 & Shared-representation linear baseline \\
Intercept softmax on CLIP & 0.5053 & 0.4782 & 0.9054 & 0.5217 & 0.1608 & Shared-representation linear variant \\
One-vs-rest logistic on CLIP & 0.5980 & 0.5871 & 0.9338 & 0.6433 & 0.1371 & Shared-representation linear variant \\
Temperature-scaled softmax on CLIP & 0.4827 & 0.4491 & 0.9286 & 0.7133 & 0.1525 & Shared-representation linear variant \\
Linear softmax head on CLIP & 0.5667 & 0.5734 & 0.9289 & 0.6608 & 0.1438 & Main-text linear representative \\
Two-layer MLP on CLIP & 0.6280 & 0.6071 & 0.9717 & 0.8333 & 0.0883 & Main-text nonlinear representative \\
CLIP-based off-the-shelf detector & 0.6773 & 0.6082 & 0.7231 & 0.2475 & 0.3367 & Main-text off-the-shelf representative \\
CNNSpot & 0.2013 & 0.1685 & 0.5108 & 0.0292 & 0.4804 & Weak off-the-shelf control \\
DIRE & 0.1900 & 0.1740 & 0.3090 & 0.0083 & 0.6554 & Weak off-the-shelf control \\
\bottomrule
\end{tabular}}
\end{table}

With the clean operating point fixed, \cref{tab:interface-appendix} is the central interface table behind \cref{tab:interface-main}. Each row changes what the attacker sees while holding the same benchmark, target direction family, attack budget, and representative verifiers fixed. The ASR columns show attack success, the $\widetilde\Gap$ columns estimate attack-derived upper bounds rather than certify positive robust separation, and the $\epsilon_{95}$ columns report validation-set emulator diagnostics. The important pattern is that high emulator error is not sufficient by itself: top-1 feedback has $\epsilon_{95}=1.0$ but can still be attacked, whereas binary feedback with abstention both removes the smooth target-aligned score and rejects borderline cases.

\begin{table}[H]
\centering
\caption{Representative interface slice at $8/255$ on the expanded five-class benchmark. This table is the detailed version of \cref{tab:interface-main}: columns report ASR, the empirical gap upper bound $\widetilde\Gap$, and the $95$th-percentile emulator error $\epsilon_{95}$. Here $q$ is a per-split transcript cap, yielding $\min(q,3000)$ training and $\min(q,1500)$ validation replies. Dashes denote model--interface combinations not evaluated in this sweep.}
\label{tab:interface-appendix}
\small
\setlength{\tabcolsep}{3pt}
\resizebox{\linewidth}{!}{%
\begin{tabular}{@{}lccc ccc ccc@{}}
\toprule
& \multicolumn{3}{c}{\textbf{Linear softmax on CLIP}} & \multicolumn{3}{c}{\textbf{Two-layer MLP on CLIP}} & \multicolumn{3}{c}{\textbf{Off-the-shelf CLIP detector}} \\
\cmidrule(lr){2-4}\cmidrule(lr){5-7}\cmidrule(l){8-10}
\textbf{Interface} & \textbf{ASR} & $\boldsymbol{\widetilde\Gap}$ & $\boldsymbol{\epsilon_{95}}$ & \textbf{ASR} & $\boldsymbol{\widetilde\Gap}$ & $\boldsymbol{\epsilon_{95}}$ & \textbf{ASR} & $\boldsymbol{\widetilde\Gap}$ & $\boldsymbol{\epsilon_{95}}$ \\
\midrule
Full score vector & 1.0000 & $-0.0356$ & 0.0051 & 1.0000 & $-0.0598$ & 0.0054 & 1.0000 & $-0.5515$ & 0.2824 \\
Single fake score & 0.9100 & $-0.0098$ & 0.3314 & 0.9300 & $-0.0200$ & 0.3137 & 1.0000 & $-0.5405$ & 0.3476 \\
8-bit fake score & 0.9233 & $-0.0100$ & 0.3196 & 0.9733 & $-0.0209$ & 0.3313 & -- & -- & -- \\
Top-1 label only & 0.7533 & $-0.0076$ & 1.0000 & 1.0000 & $-0.0380$ & 1.0000 & 1.0000 & $-0.5512$ & 1.0000 \\
Binary verdict & 1.0000 & $-0.0260$ & 0.7056 & 1.0000 & $-0.0498$ & 0.7004 & -- & -- & -- \\
Binary + abstain ($q=128$) & 0.0000 & $+0.0137$ & 1.0000 & 0.0000 & $+0.0104$ & 1.0000 & 0.0300 & $+0.2858$ & 1.0000 \\
Binary + abstain ($q=7{,}200$) & 0.0000 & $+0.0104$ & 1.0000 & 0.0000 & $+0.0120$ & 1.0000 & 0.0933 & $+0.2165$ & 1.0000 \\
\bottomrule
\end{tabular}}
\end{table}

\Cref{fig:interface-appendix} gives the same comparison as a visual diagnostic. It is included to make the three-way pattern easier to inspect: emulator error, ASR, and $\widetilde\Gap$ move differently across interfaces, so a single metric would obscure the reason that abstention helps.

\begin{figure}[H]
\centering
\IfFileExists{fig_interface_bandwidth.png}{%
\includegraphics[width=0.99\linewidth]{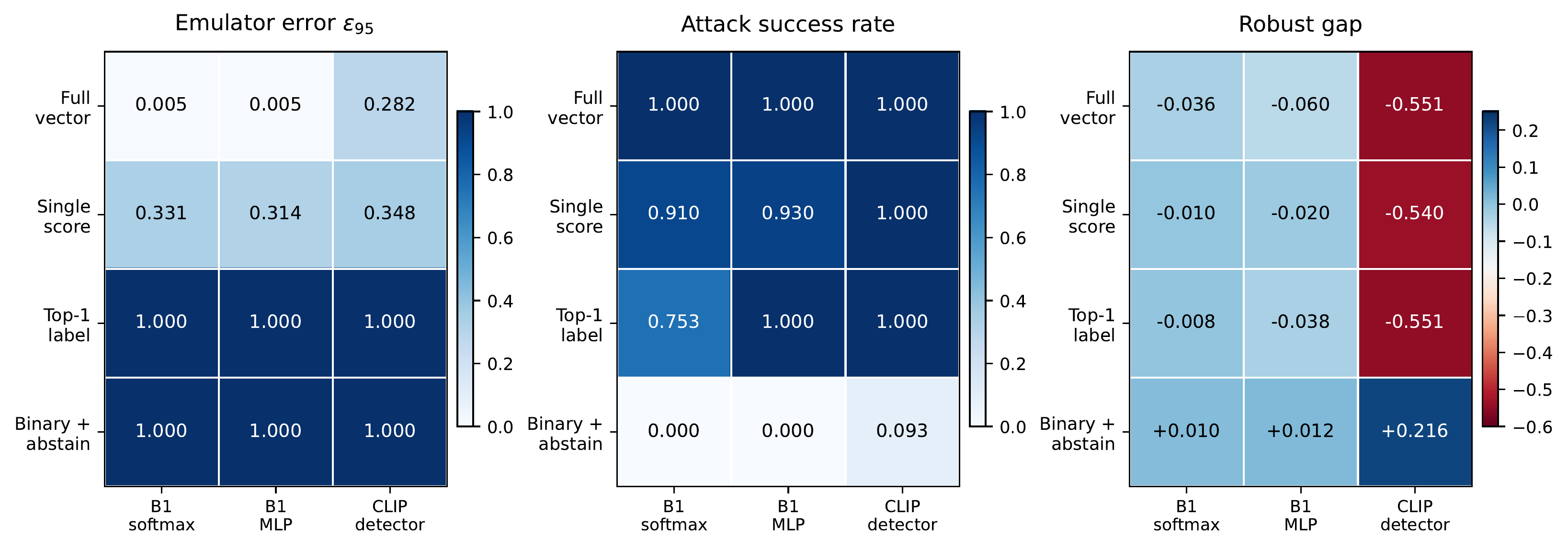}%
}{%
\fbox{\parbox{0.92\linewidth}{\centering Placeholder for \texttt{fig\_interface\_bandwidth.pdf}. The figure summarizes emulator error, ASR, and empirical gap upper bounds across interface bandwidths.}}%
}
\caption{Heatmap summary of the representative $8/255$ interface slice. This visual counterpart to \cref{tab:interface-appendix} shows that richer interfaces remain fragile and that positive empirical gap upper bounds first appear only under binary feedback with abstention.}
\Description{A three-panel heatmap of emulator error, ASR, and empirical gap upper bounds under different interfaces.}
\label{fig:interface-appendix}
\end{figure}

\Cref{tab:binary-query-appendix} varies the availability of binary transcripts without abstention. ASR remains $1.0$ and estimated gap upper bounds remain negative across the tested per-split caps. These are controlled transcript-availability measurements, not end-to-end commercial API query budgets. They show that ordinary binary feedback remains attackable under the tested caps.

\begin{table}[H]
\centering
\caption{Binary-feedback transcript-cap sweep for the linear softmax CLIP head on the expanded five-class benchmark. The cap is reset for each split; the actual reply counts are shown separately. Without abstention, ASR remains $1.0$ and estimated gap upper bounds remain negative.}
\label{tab:binary-query-appendix}
\small
\begin{tabular}{@{}rrrrrr@{}}
\toprule
\textbf{Cap $q$} & \textbf{Train replies} & \textbf{Val. replies} & \textbf{Binary ASR} & $\boldsymbol{\widetilde\Gap}$ & $\boldsymbol{\epsilon_{95}}$ \\
\midrule
128 & 128 & 128 & 1.0000 & $-0.0249$ & 0.6846 \\
512 & 512 & 512 & 1.0000 & $-0.0253$ & 0.7159 \\
2{,}048 & 2{,}048 & 1{,}500 & 1.0000 & $-0.0263$ & 0.7098 \\
7{,}200 & 3{,}000 & 1{,}500 & 1.0000 & $-0.0258$ & 0.7025 \\
\bottomrule
\end{tabular}
\end{table}

Finally, \cref{tab:private-oracle-appendix} reports a score-guided random-search sanity check. The search uses the deployed continuous target score to retain the best proposal and reports threshold-crossing success only; no $\widetilde\Gap$ was computed. It is not a purely binary private-oracle experiment and does not validate the hidden-state query bound in \cref{thm:hidden-state}.

\begin{table}[H]
\centering
\caption{Score-guided random-search sanity check. Entries are threshold-crossing success rates when continuous target scores guide proposal retention. This control reports neither a robust-gap estimate nor binary-only private-oracle performance.}
\label{tab:private-oracle-appendix}
\small
\begin{tabular}{@{}rcc@{}}
\toprule
\textbf{Queries} & \textbf{Linear softmax} & \textbf{Two-layer MLP} \\
\midrule
128 & 0.3438 & 0.1562 \\
512 & 0.4062 & 0.1875 \\
2{,}048 & 0.4062 & 0.2500 \\
7{,}200 & 0.4062 & 0.3125 \\
\bottomrule
\end{tabular}
\end{table}

\subsection{Adversarially trained and stronger clean three-class models}
\label{subsec:app-retrain-controls}

The last appendix block addresses the clean-strength objection raised in \cref{subsec:exp-retrain-control}. \Cref{tab:advd1-attack-appendix} gives the per-direction attack results for the adversarially trained three-class model. The table is needed because an aggregate ASR can hide whether one direction remains protected; here, every direction has negative estimated gap upper bounds under the evaluated budgets.

\begin{table}[H]
\centering
\caption{Per-direction targeted attacks for the adversarially trained three-class model. This table supports \cref{subsec:exp-retrain-control}: relative to the released checkpoint, clean performance improves substantially, but the model still fails under the same source--target attack directions.}
\label{tab:advd1-attack-appendix}
\normalsize
\setlength{\tabcolsep}{5pt}
\begin{tabular}{@{}llrrrr@{}}
\toprule
\textbf{Direction} & \textbf{Budget} & \textbf{ASR} & $\boldsymbol{\widetilde\Gap}$ & \shortstack[c]{\textbf{Clean target}\\\textbf{acc.}} & \shortstack[c]{\textbf{Adv. target}\\\textbf{score}} \\
\midrule
Human$\to$AI & $2/255$ & 1.0000 & $-0.0576$ & 0.6807 & 0.7383 \\
Human$\to$AI & $4/255$ & 1.0000 & $-0.0938$ & 0.6807 & 0.7745 \\
Human$\to$AI & $8/255$ & 1.0000 & $-0.1236$ & 0.6807 & 0.8042 \\
AI$\to$Human & $2/255$ & 0.9000 & $-0.0483$ & 0.3730 & 0.4213 \\
AI$\to$Human & $4/255$ & 0.9817 & $-0.0809$ & 0.3730 & 0.4539 \\
AI$\to$Human & $8/255$ & 1.0000 & $-0.1141$ & 0.3730 & 0.4872 \\
SD1.5$\to$SDXL & $2/255$ & 1.0000 & $-0.0694$ & 0.3899 & 0.4594 \\
SD1.5$\to$SDXL & $4/255$ & 1.0000 & $-0.1305$ & 0.3899 & 0.5204 \\
SD1.5$\to$SDXL & $8/255$ & 1.0000 & $-0.1847$ & 0.3899 & 0.5747 \\
SDXL$\to$SD1.5 & $2/255$ & 0.9267 & $-0.0238$ & 0.3623 & 0.3861 \\
SDXL$\to$SD1.5 & $4/255$ & 0.9933 & $-0.0421$ & 0.3623 & 0.4044 \\
SDXL$\to$SD1.5 & $8/255$ & 1.0000 & $-0.0615$ & 0.3623 & 0.4238 \\
\bottomrule
\end{tabular}
\end{table}

Across five seeds, full-batch retraining yields test accuracy $0.570\pm0.047$ and fake AUC $0.905\pm0.021$, whereas the stronger clean-only model reaches $0.870\pm0.003$ and $0.9823\pm0.0002$. The representative clean-only model used in \cref{subsec:exp-retrain-control} reaches clean accuracy $0.8689$ and fake AUC $0.9822$. Its targeted-attack evaluation still yields minimum ASR $0.9667$ at $2/255$, ASR $1.0$ for all four directions by $4/255$, and non-positive empirical gap upper bounds throughout. Together with \cref{tab:advd1-attack-appendix}, these controls support the main-text conclusion that clean-strength improvements alone do not remove the interface-learnability failure mode.

\section{Empirical Bracketing of Barrier I}
\label{app:sandwich}

This appendix connects the minimax statement in \cref{sec:barrier1} to the empirical quantities reported in \cref{sec:exp} and Appendix~\ref{app:exp-details}. The attackability radius itself is difficult to estimate for realistic images. The following inequalities describe a conceptual bracket; obtaining tight computable bounds for realistic images remains unresolved. The image experiments instead estimate attack-derived upper bounds for fixed verifiers, subject to the finite-sample interpretation in \cref{sec:exp}. By \cref{thm:attackability-law}, for any verifier family $\mathcal{F}\subseteq[0,1]^\X$,
\[
    \sup_{f\in\mathcal{F}}\Gap_{s,t}(f)
    \le
    \AttackRad_{s,t}.
\]
Thus the true robust gap of any fixed deployed verifier, evaluated against the full admissible attack class, lower-bounds the ceiling. An evaluated heuristic attack only supplies an upper bound on the fixed-verifier gap through $\widetilde\Gap$; it does not thereby lower-bound $\AttackRad_{s,t}$.

On the attack side, if $\widetilde{\Qset}_s\subseteq\Qset_s$ is any explicit parametric family of attacked source distributions, then
\[
    \AttackRad_{s,t}
    =
    \inf_{Q\in\Qset_s}\TV(P_t,Q)
    \le
    \inf_{Q\in\widetilde{\Qset}_s}\TV(P_t,Q).
\]
This is why explicit editors can upper-bound the statistical ceiling when their induced distributions can be compared to the target. Our image experiments do not estimate these bounds tightly enough to recover $\AttackRad_{s,t}$. Appendix~\ref{app:additional-finite-state} separately computes both minimax sides in a finite-state setting with known distributions.

\section{Extension Beyond Finite Image Spaces}
\label{app:continuous-route}

\Cref{sec:model} assumes a finite image space to keep the main proof elementary and to match quantized digital images. The attackability-radius law itself is not tied to finiteness: it extends to more general measurable spaces under compactness in total-variation norm. This section records that extension so that Barrier~I is understood as a distributional statement rather than an artifact of finite notation.

\begin{proposition}[Barrier I beyond finite image spaces]
\label{prop:tv-compact-extension}
Let $(\X,\mathcal{F})$ be a measurable space and let $\mathcal{P}(\X)$ denote its probability measures. Fix a target distribution $P_t\in\mathcal{P}(\X)$ and let $\Qset_s\subseteq\mathcal{P}(\X)$ be nonempty, convex, and compact in total-variation norm. Then
\[
\sup_{f:\X\to[0,1]\ \mathrm{measurable}}
\left(\E_{P_t}[f]-\sup_{Q\in\Qset_s}\E_Q[f]\right)
=
\inf_{Q\in\Qset_s}\TV(P_t,Q).
\]
\end{proposition}

\begin{proof}
Let $\mathcal{B}_b(\X)$ be the Banach space of bounded measurable real-valued functions on $\X$ with the sup norm, and let
\[
\mathcal{F}_{[0,1]}:=\{f\in\mathcal{B}_b(\X):0\le f\le1\}.
\]
Define
\[
\Lambda(Q,f):=\E_{P_t}[f]-\E_Q[f].
\]
The set $\mathcal{F}_{[0,1]}$ is convex, and $\Qset_s$ is convex and compact by assumption. For fixed $Q$, the map $f\mapsto\Lambda(Q,f)$ is affine and continuous in sup norm. For fixed $f$, the map $Q\mapsto\Lambda(Q,f)$ is affine and continuous in total-variation norm. Sion's minimax theorem therefore yields
\[
\sup_{f\in\mathcal{F}_{[0,1]}}\inf_{Q\in\Qset_s}\Lambda(Q,f)
=
\inf_{Q\in\Qset_s}\sup_{f\in\mathcal{F}_{[0,1]}}\Lambda(Q,f).
\]
For fixed $Q$, the inner supremum is the standard variational characterization of total variation on measurable spaces, so the result follows.
\end{proof}

\clearpage
\section{Reproducibility, Societal Impact, and AI-Assisted Writing}
\label{app:open-science}

The public source repository is
\url{https://github.com/kaikaiyao/pixels-alone-provenance}.
It contains core experiment code and configurations, including implementations of the finite-state, RGB ResNet-18, and paired image-quality evaluations in Appendix~\ref{app:additional}. Datasets, model weights, and saved experimental outputs are not bundled. Reproducing the image experiments requires the corresponding external data and model assets.

\paragraph{Compute resources.}
The original image experiments used NVIDIA A100 GPU workers. The main controlled studies use frozen CLIP features and small verifier heads rather than end-to-end training of large generative models. The exact-recovery, interface-sweep, binary-feedback, score-guided random-search, and retraining controls are dominated by feature extraction, small-head training, and attack/evaluation loops. The additional ResNet-18 victim is trained from RGB pixels, and the finite-state study solves linear programs. No reported experiment requires training a new image generator or a large vision-language backbone from scratch.

\paragraph{Experimental uncertainty.}
The image attack results are descriptive evaluations on held-out samples; their gap estimates are not accompanied by population confidence bounds. The original stochastic retraining controls report mean and standard deviation over five seeds. The finite-state plug-in study uses 20 repeats per setting and reports mean absolute error against known population values. The additional ResNet-18 evaluation uses one seed; its generality and training variability remain unresolved.

\paragraph{Existing assets and licenses.}
The experiments use existing assets credited in the paper and source bundle: MS COCO prompts/images (COCO annotations under CC-BY 4.0, with images subject to their original Flickr licenses and terms), OpenAI CLIP (MIT-licensed repository and public model-release terms), Stable Diffusion 1.5 (CreativeML OpenRAIL-M), Stable Diffusion 2.1, SDXL, and PixArt-family generators (CreativeML Open RAIL++-M/OpenRAIL++ model-card terms), and detector baselines from their cited authors. The source release contains code and configurations; datasets, image assets, model checkpoints, and saved experimental outputs are not redistributed. Users must obtain upstream assets from the original providers under their original terms.

\paragraph{New assets.}
The new assets are experiment source code and configuration files, including implementations of the additional evaluations in Appendix~\ref{app:additional}. Repository documentation describes the code and required external inputs; no new dataset or model checkpoint is released.

\paragraph{Human subjects and IRB.}
The work does not involve crowdsourcing, participant studies, intervention with human subjects, or collection of private user data. No IRB review was sought because the experiments use public prompts/images, synthetic generations, and offline research implementations.

This paper studies attacks on passive provenance systems and therefore has dual-use implications. Targeted-forgery techniques could be adapted to weaken real provenance or moderation pipelines. The analysis is defensive: it clarifies that clean provenance accuracy can be a misleading robustness signal and that public verifier interfaces can become a dominant source of fragility. All experiments were conducted offline on controlled benchmarks and research implementations rather than against live moderation systems or third-party services.

A generative AI assistant was used for language polishing and grammar editing during manuscript preparation, and coding assistants were used to help implement and monitor empirical experiments. The authors are responsible for verifying all code, results, claims, and citations prior to submission.

\clearpage
\section{Additional Protocols and Evaluations}
\label{app:additional}

\subsection{Surrogate training and attack protocol}
\label{app:additional-protocol}

The controlled benchmark contains aligned real, Stable Diffusion~1.5, and SDXL images for 1{,}200 MS COCO prompts, split by prompt into 600/300/300 training/validation/test prompts.
The expanded benchmark adds Stable Diffusion~2.1 and PixArt under the same split protocol.
Training and validation interface transcripts remain separated by prompt.
For exact softmax-head recovery, the attacker queries full probability vectors on public CLIP features, converts the replies to log-odds, and solves an overdetermined least-squares system for the relative head parameters.
For learned emulators, we use a two-layer model over public features with a 256-unit GELU hidden layer, trained for 30 full-batch epochs using AdamW with learning rate $10^{-3}$ and weight decay $10^{-4}$.
The loss follows the released response: mean-squared error for probability vectors or scalar scores, binary cross-entropy for binary responses, and cross-entropy for label-only or abstaining responses.

On the expanded benchmark, a transcript cap $q$ applies separately to each split, yielding $\min(q,3000)$ training replies and $\min(q,1500)$ validation replies.
Thus a cap of 7{,}200 yields 4{,}500 replies in total for this interface sweep; it is distinct from the 7{,}200 probes in the exact-recovery experiment.
These controlled interface wrappers are API-style abstractions that isolate the effect of released feedback.
They do not evaluate a live commercial service, whose preprocessing, model updates, account controls, and query limits may differ.

After fitting the emulator, the attacker differentiates through the public feature extractor and recovered or learned head to maximize the chosen target score.
The pixel attacks use targeted projected gradient descent (PGD) under $\ell_\infty$ budgets $\delta\in\{2/255,4/255,8/255\}$, with 40 steps, step size $\delta/4$, and five restarts.
The deployed verifier then scores each returned image for the reported ASR and empirical gap estimate.
The validation-set $95$th-percentile emulator error, $\epsilon_{95}$, measures held-out approximation error; it does not bound the uniform error $\varepsilon$ over the full attack region required by \cref{thm:surrogate-attack}.
The interface sweep also does not estimate the average surrogate optimization error $\bar\alpha$ separately, so its diagnostics cannot be substituted into the theorem's $2\varepsilon+\bar\alpha$ guarantee.

\subsection{Finite-state validation of Barrier I}
\label{app:additional-finite-state}

We evaluate the minimax identity on $\X=\{0,1\}^{6}$, where distributions and reachable edits can be represented explicitly.
The study uses five fixed source--target distribution pairs and Hamming edit sets $\Gamma_r(x)=\{z:d_H(x,z)\le r\}$ for $r\in\{0,1,2,3\}$, giving 20 population settings.
For each setting, an attack-side transport linear program and an independently formulated verifier-side linear program compute opposite sides of \cref{thm:attackability-law}.

The two programs can be written for any finite source and target mass functions $p_s$ and $p_t$ as follows.
Let $\pi_{xz}$ be the probability mass transported from source state $x$ to output state $z$, and let $q_z=\sum_x\pi_{xz}$.
The attack-side program is
\begin{equation}
\begin{aligned}
\min_{\pi,q,v}\quad & \frac{1}{2}\sum_{z\in\X}v_z \\
\text{subject to}\quad
& \pi_{xz}\ge0,\qquad \sum_{z\in\X}\pi_{xz}=p_s(x), \\
& \pi_{xz}=0\quad\text{if }z\notin\Gamma_r(x),\qquad q_z=\sum_{x\in\X}\pi_{xz}, \\
& v_z\ge p_t(z)-q_z,\qquad v_z\ge q_z-p_t(z).
\end{aligned}
\label{app:additional-transport-lp}
\end{equation}
Its optimum is $\inf_{Q\in\Qset_s}\TV(P_t,Q)$.
The verifier-side program introduces $u_x$ to upper-bound the largest admissible target score for source state $x$:
\begin{equation}
\begin{aligned}
\max_{f,u}\quad & \sum_{z\in\X}p_t(z)f_z-\sum_{x\in\X}p_s(x)u_x \\
\text{subject to}\quad
& 0\le f_z\le1,\qquad 0\le u_x\le1, \\
& u_x\ge f_z\quad\text{for every }z\in\Gamma_r(x).
\end{aligned}
\label{app:additional-verifier-lp}
\end{equation}
Its optimum is $\sup_f\Gap_{s,t}(f)$ because an optimal $u_x$ can be chosen as the largest score among the admissible edits.

Across the 20 population settings, the maximum absolute discrepancy between the two independently computed values is $4.44\times10^{-16}$.
We also replace the source and target distributions with empirical histograms and compare the resulting plug-in radius with the known population value, using 20 repeats per setting.
\Cref{tab:finite-state} reports the mean absolute error across these settings and repeats.
The decreasing error checks estimation in this finite domain; it does not estimate $\AttackRad_{s,t}$ for real high-dimensional image distributions.

\begin{table}[htbp]
\centering
\caption{Finite-state plug-in estimation of the attackability radius. Mean absolute error is aggregated over 20 source--target/edit-radius settings and 20 repeats per setting. The population attack-side and verifier-side optima agree to a maximum absolute discrepancy of $4.44\times10^{-16}$.}
\label{tab:finite-state}
\begin{tabular}{rr}
\toprule
\textbf{Sample size $n$} & \textbf{Mean absolute error} \\
\midrule
250 & 0.02094 \\
1{,}000 & 0.01008 \\
5{,}000 & 0.00451 \\
\bottomrule
\end{tabular}
\end{table}

\subsection{Transfer to a non-CLIP victim}
\label{app:additional-resnet}

To test whether transfer requires CLIP features in the victim, we evaluate an ImageNet-pretrained ResNet-18~\cite{he2016deep} trained end-to-end on RGB pixels and fixed during attack evaluation.
It reaches validation/test fake AUCs of 0.9830/0.9862 and test balanced accuracy of 93.58\% at the validation-selected threshold of 0.8480.
The victim returns one continuous fake-class score per query.
The transcript contains 9{,}000 scalar replies: a clean image and one uniformly perturbed $\ell_\infty\le8/255$ variant for each of 3{,}000 training and 1{,}500 validation images.
The attacker fits a two-layer emulator over frozen CLIP ViT-B/32 features and performs targeted PGD at $8/255$, with 40 steps and five restarts, without adaptive victim queries during attack optimization.

\Cref{tab:resnet-transfer} separates overall ASR from success among sources classified correctly before attack.
In the fake-to-real direction, 45 of 117 initially correct sources become accepted as real.
In the reverse direction, none of the 122 initially correct sources succeeds; the sole overall success was already classified as the target before attack.
The positive empirical gap estimates leave the sign of the true robust gap unresolved.
These results support partial directional transfer across victim and emulator representations, under one victim architecture, one seed, the same-prompt benchmark, and the evaluated bounded pixel edits.

\begin{table}[htbp]
\centering
\small
\caption{Transfer from a CLIP-feature emulator to an RGB ResNet-18 victim at $\ell_\infty=8/255$. Each direction contains 128 held-out sources. The initially-correct column excludes clean misclassifications and reports the number of induced target predictions. Positive empirical gap estimates do not certify positive robust separation.}
\label{tab:resnet-transfer}
\begin{tabular}{lrrr}
\toprule
\textbf{Direction} & \textbf{Overall ASR} & \shortstack{\textbf{Initially-correct}\\\textbf{ASR (count)}} & $\boldsymbol{\widetilde\Gap}$ \\
\midrule
Fake $\to$ Real & 43.75\% & 38.46\% (45/117) & $+0.5375$ \\
Real $\to$ Fake & 0.78\% & 0.00\% (0/122) & $+0.9093$ \\
\bottomrule
\end{tabular}
\end{table}

\subsection{Paired image quality and attack success}
\label{app:additional-quality}

We measure structural similarity (SSIM)~\cite{wang2004ssim}, learned perceptual image patch similarity with an AlexNet backbone (LPIPS-Alex)~\cite{zhang2018unreasonable}, and peak signal-to-noise ratio (PSNR) between clean and attacked images.
Higher SSIM and PSNR, and lower LPIPS, indicate closer agreement under the respective paired metrics.
This evaluation uses 16 held-out sources for each of the four attack directions in the controlled benchmark, at both $2/255$ and $8/255$.
There are 64 source--target instances per budget and 128 attack instances in total.
At each budget, 14 instances already have the target-side prediction before attack, leaving 50 initially correct instances for measuring induced success.

\Cref{tab:paired-quality} reports induced success and ranges of per-direction median quality over all 16 instances in each direction.
At $2/255$, the attack induces the target in 45 of 50 initially correct instances.
Restricting the quality calculation to these 45 induced successes gives pooled medians of 0.9784 SSIM, 0.0063 LPIPS-Alex, and 43.82\,dB PSNR.
Increasing the budget to $8/255$ induces five additional successes, with lower SSIM and higher LPIPS medians.
Thus the larger budget trades lower paired image quality for higher induced success.
No evaluated instance exceeds its $\ell_\infty$ budget under the prespecified $10^{-6}$ tolerance.
These metrics quantify fidelity to the source image; a claim about human perception or semantic preservation would require a separate evaluation.

\begin{table}[htbp]
\centering
\small
\caption{Paired image quality on the controlled benchmark. Each budget contains 64 attack instances, including 14 pre-existing target-side predictions. Induced ASR uses the remaining 50 initially correct instances. SSIM and LPIPS-Alex ranges are the minimum and maximum of the four direction-specific medians, each calculated over all 16 instances in that direction, rather than over successful attacks alone.}
\label{tab:paired-quality}
\begin{tabular}{lrrr}
\toprule
\textbf{Budget} & \shortstack{\textbf{Induced ASR}\\\textbf{(count)}} & \shortstack{\textbf{Median SSIM}\\\textbf{range}} & \shortstack{\textbf{Median LPIPS-Alex}\\\textbf{range}} \\
\midrule
$2/255$ & 90.0\% (45/50) & 0.9781--0.9849 & 0.0050--0.0081 \\
$8/255$ & 100.0\% (50/50) & 0.8011--0.8534 & 0.0770--0.1365 \\
\bottomrule
\end{tabular}
\end{table}

\Cref{fig:quality-eps2,fig:quality-eps8} show selected clean/adversarial pairs in all four directions using the same clean images at both budgets.
The $2/255$ examples include an attack failure, while the selected $8/255$ examples induce the target prediction.
Signed perturbations are displayed with the original eightfold amplification around mid-gray; these visualizations illustrate individual outcomes rather than estimate their prevalence.

\begin{figure}[p]
\centering
\includegraphics[height=0.80\textheight,keepaspectratio]{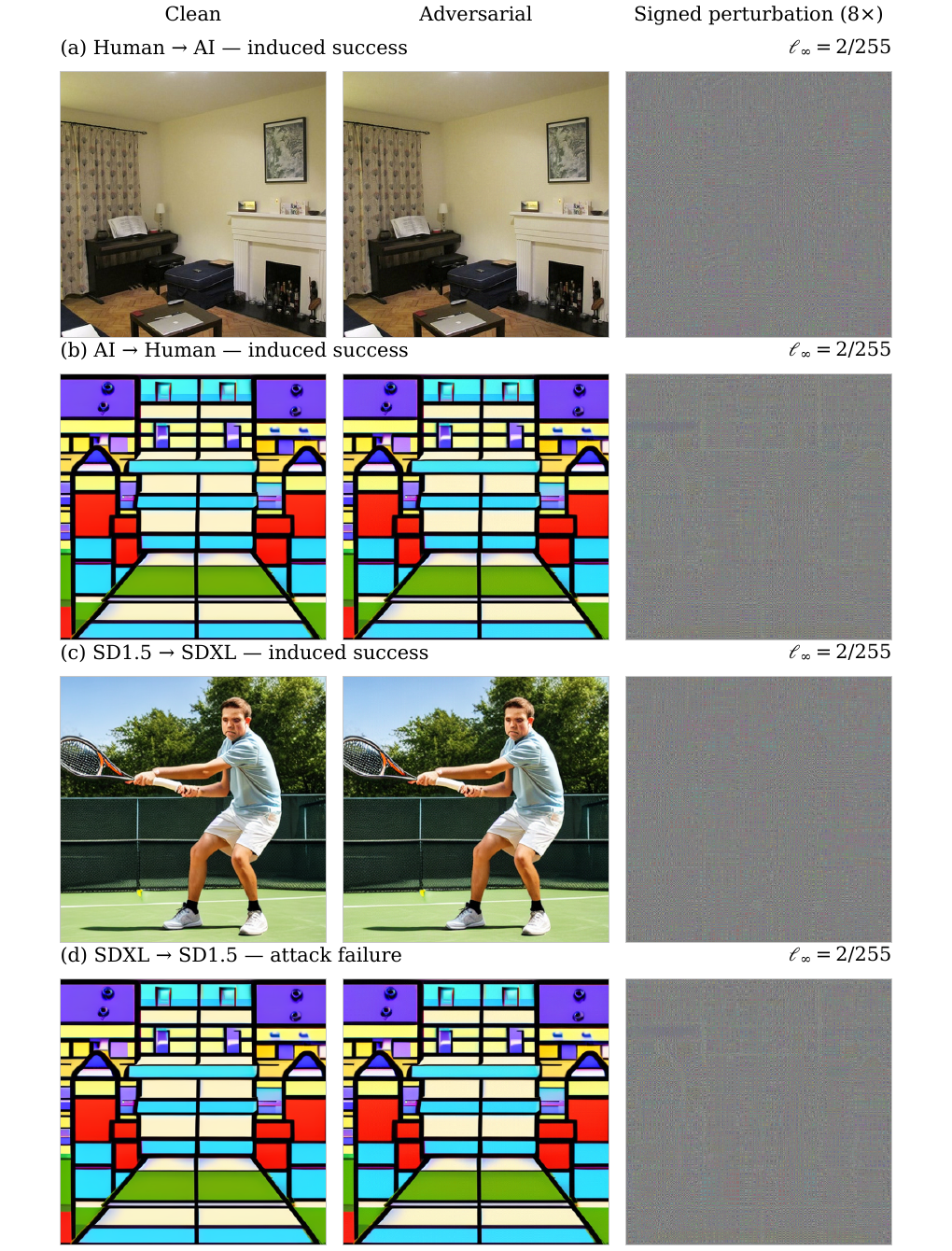}
\caption{Selected paired examples in four attack directions at an $\ell_\infty$ budget of $2/255$. Columns show the clean image, the attacked image, and the signed perturbation amplified eightfold around mid-gray. The first three rows are induced successes and the SDXL$\to$SD1.5 row is an attack failure. Figure~\ref{fig:quality-eps8} uses the same clean images at $8/255$.}
\label{fig:quality-eps2}
\end{figure}

\begin{figure}[p]
\centering
\includegraphics[height=0.80\textheight,keepaspectratio]{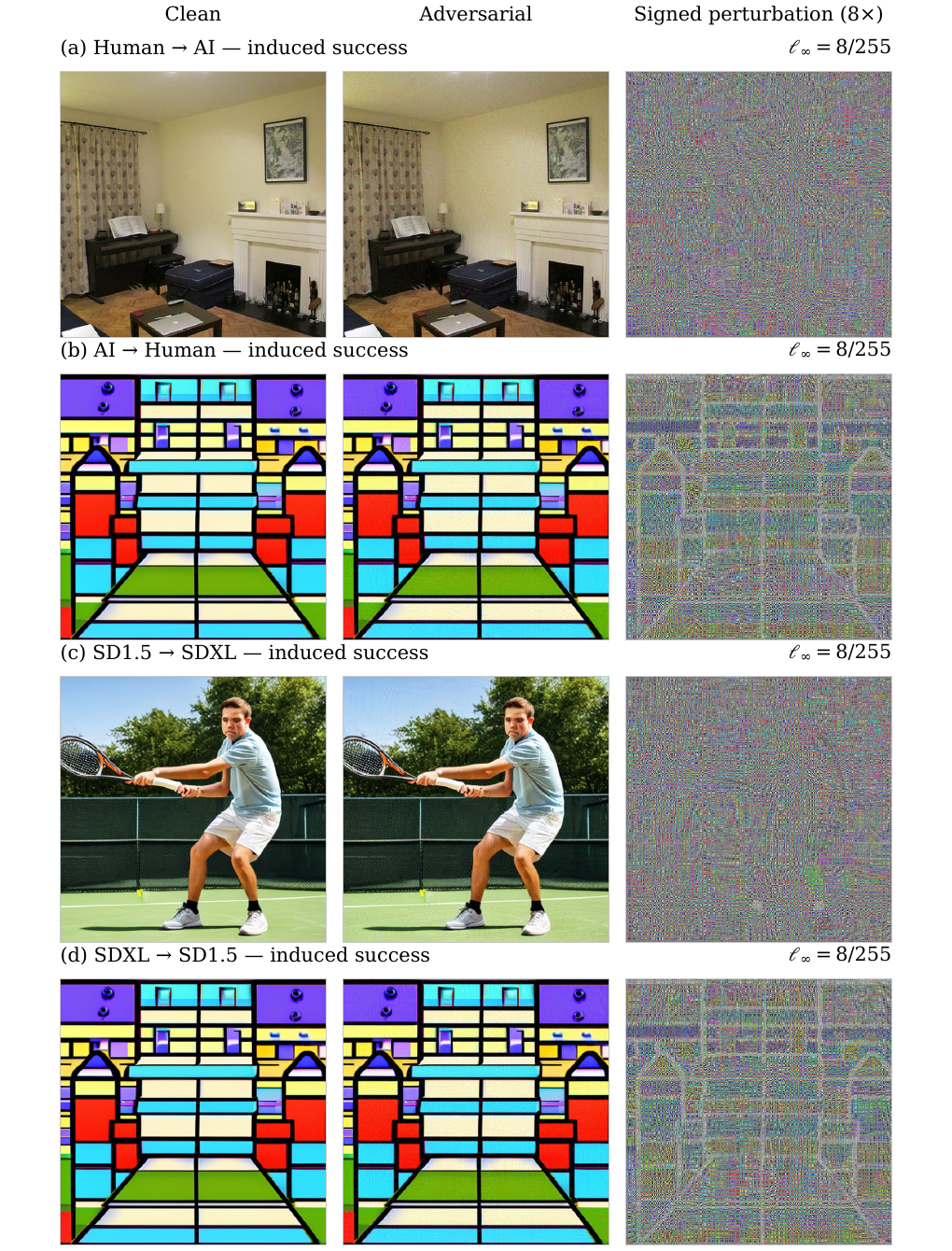}
\caption{The same four clean images as in Figure~\ref{fig:quality-eps2}, attacked at an $\ell_\infty$ budget of $8/255$. All four selected examples are induced successes. The paired views and amplified signed perturbations illustrate the larger changes at this budget; they do not establish human imperceptibility or semantic preservation.}
\label{fig:quality-eps8}
\end{figure}


\end{document}